\documentclass[11pt]{article}

\usepackage[margin=1in]{geometry}
\usepackage{amsmath,amssymb,amsthm,mathtools}
\usepackage[hidelinks]{hyperref}
\usepackage{amsmath, mdframed}
\usepackage{amsfonts}
\usepackage{newtxtext, newtxmath}
\usepackage{enumitem}
\usepackage{titling}
\usepackage{epsfig}
\usepackage{hyperref}
\usepackage{url}
\usepackage{graphicx}
\usepackage{nicematrix}
\usepackage[utf8]{inputenc}
\usepackage[english]{babel}
\usepackage{CJKutf8}
\usepackage{comment}
\usepackage{bbm}
\usepackage{blindtext}
\usepackage{xcolor}
\usepackage{tikz, tikz-3dplot}
\usepackage{tikz-qtree}
\usetikzlibrary{positioning,calc}
\usepackage{etoolbox}
\usepackage{enumitem}
\usepackage{multirow}
\usepackage{tabularx}
\usepackage{makecell}
\usepackage[nocompress,sort]{cite}
\usepackage{algorithm}
\usepackage{algpseudocode}
\usetikzlibrary{decorations.markings}
\usetikzlibrary{decorations.text}
\usetikzlibrary{arrows.meta}
\hypersetup{
    colorlinks=true,       
    linkcolor=blue,        
    citecolor=magenta,         
    filecolor=magenta,     
    urlcolor=cyan,         
    linktoc=all
}

\usepackage{xcolor}

\definecolor{mapink}{HTML}{172033}
\definecolor{mapmuted}{HTML}{667085}
\definecolor{mapopen}{HTML}{2563B8}
\definecolor{mapfalse}{HTML}{D92D20}
\definecolor{maptheorem}{HTML}{000000}

\newtheorem{theorem}{Theorem}

\newtheorem{lemma}[theorem]{Lemma}
\newtheorem{corollary}[theorem]{Corollary}
\newtheorem{prop}[theorem]{Proposition}

\newtheorem{conjecture}{Conjecture}

\newcommand{\ind}{\mathbf 1}
\newcommand{\cF}{\mathcal F}
\newcommand{\cH}{\mathcal H}

\newcommand{\cI}{\mathcal I}
\newcommand{\cB}{\mathcal B}

\newcommand{\inftynorm}[1]{\left\lVert#1\right\rVert_\infty}

\newcommand{\onenorm}[1]{\left\lVert#1\right\rVert_1}
\newcommand{\rounddown}[1]{\left\lfloor#1\right\rfloor}
\newcommand{\roundup}[1]{\left\lceil#1\right\rceil}
\newcommand{\Z}{\mathbb{Z}}
\newcommand{\R}{\mathbb{R}}

\DeclareMathOperator{\spa}{span}
\DeclareMathOperator{\opt}{OPT}

\def\final{0}  
\def\iflong{\iffalse}
\ifnum\final=0  
\newcommand{\kristof}[1]{{\color{red}[{\tiny \textbf{Kristóf:}  #1}]\marginpar{\color{red}*}}}
\newcommand{\siyue}[1]{{\color{blue}[{\tiny \textbf{Siyue:}  #1}]\marginpar{\color{blue}*}}}
\newcommand{\jakub}[1]{{\color{green}[{\tiny \textbf{Jakub:}  #1}]\marginpar{\color{green}*}}}
\newcommand{\victor}[1]{{\color{teal}[{\tiny \textbf{Victor:}  #1}]\marginpar{\color{teal}*}}}
\else 
\newcommand{\kristof}[1]{}
\newcommand{\siyue}[1]{}
\newcommand{\victor}[1]{}
\newcommand{\jakub}[1]{}
\fi

\newlength{\bibitemsep}
\newlength{\bibparskip}
\let\oldthebibliography\thebibliography
\renewcommand\thebibliography[1]{%
  \oldthebibliography{#1}%
  \setlength{\parskip}{\bibitemsep}%
  \setlength{\itemsep}{\bibparskip}%
}

\makeatletter
\renewcommand{\paragraph}{%
  \@startsection{paragraph}{4}%
  {\z@}{1.6ex \@plus 1ex \@minus .2ex}{-0.5em}%
  {\normalfont\normalsize\bfseries}%
}
\makeatother

\title{Weighted Equitability and Matroid-Constrained Discrepancy}
\author{Krist\'of B\'erczi\thanks{MTA-ELTE Matroid Optimization Research Group and HUN-REN-ELTE Egerváry Research Group, Department of Operations Research, ELTE Eötvös Loránd University, and HUN-REN Alfréd Rényi Institute of Mathematics, Budapest, Hungary. Email: \texttt{kristof.berczi@ttk.elte.hu}.}
\and Siyue Liu\thanks{Carnegie Mellon University, Pittsburgh, PA, USA. Email: \texttt{siyueliu@andrew.cmu.edu}.}
\and Victor Reis\thanks{Microsoft Research, Redmond, WA, USA. Email: \texttt{victorol@microsoft.com}.}
\and Jakub Tarnawski\thanks{Microsoft Research, Redmond, WA, USA. Email: \texttt{jatarnaw@microsoft.com}.}}

\date{}

\begin{document}

\maketitle

\begin{abstract}
    We prove weighted matroid equitability. Let $M=(E,\cI)$ be a matroid whose ground set can be partitioned into $k$ bases, and assign a nonnegative weight to every element. Then $E$ has a partition into $k$ bases such that the weights of any two bases differ by at most the largest element weight.
    We present two proofs based on the localized exchange theorem of Akrami, Liu, Raj, and V\'egh. The first is existential, while the second is constructive and leads to a strongly polynomial-time algorithm.
    As applications, we obtain an additive guarantee for matroid-constrained makespan minimization for identical machines and a strongly polynomial-time algorithm for finding an EF1 allocation under a matroid constraint and identical additive valuations.
    
    We further generalize the Beck--Fiala framework in discrepancy theory to settings with matroid constraints. Given a nonnegative matrix of column sparsity $\Delta$, we show that a fractional basis can be rounded to a basis of no larger cost while preserving every row sum within additive error $2\Delta$ times the largest matrix entry. Motivated by the $2$-sparse prefix Beck--Fiala conjecture, we formulate a conjecture on prefix-constrained matroid bases and prove a discrepancy bound $O(\log n)$. Finally, we give a counterexample to the weighted carpooling conjecture, thereby also disproving a conjecture by Morell and Skutella on single-source unsplittable flows with two-sided discrepancy bounds.
    
\end{abstract}


\section{Introduction}
\label{sec:intro}

Given a fractional point in the intersection of (multiple) matroid base polytopes and a system of linear equations defined by a matrix $W$, can we find a matroid basis that has small violation of the linear equations? We study closely related questions on matroid equitability, degree-bounded bases, prefix-constrained bases, and unsplittable flows. The common theme of those problems is that the unweighted case where the entries of $W$ are $\{0,1\}$ is resolved, whereas the weighted versions are seemingly much harder. We study those questions from two angles, matroid equitability and matroid-constrained discrepancy.

A basic question in matroid partitions asks how evenly a distinguished set can be distributed among disjoint bases. Fekete and Szabó~\cite{fekete2011equitable} conjectured that whenever the ground set of a matroid can be partitioned into $k$ bases, then for every subset $S$ there is such a partition in which the numbers of elements of $S$ in the bases differ by at most one. Akrami, Liu, Raj, and Végh~\cite{AkramiLiuRajVegh2026} recently proved this conjecture for all matroids and gave a polynomial-time algorithm.

We consider the weighted analogue. Given a matroid $M=(E,\mathcal{I})$ and nonnegative weights $w\in \R_{\ge 0}^E$, can the ground set be partitioned into bases whose total weights differ by at most the largest element weight? Our main result answers this question affirmatively: there is a partition into bases $B_1,\dots,B_k$ such that
\begin{equation}  
    |w(B_i)-w(B_j)|\leq\max_{e\in E}w_e \label{eq:weighted-equitability}
\end{equation}
for all $i,j\in[k]$; such a partition is called \emph{equitable}. 
\begin{theorem}[Weighted matroid equitability]\label{thm:weighted}
Let $M=(E,\cI)$ be a matroid whose ground set can be partitioned into $k$ bases, and let $w\in\R_{\ge0}^E$. Then $E$ has an equitable partition into $k$ bases $E=B_1\sqcup\cdots\sqcup B_k$ satisfying \eqref{eq:weighted-equitability}. Consequently,
\begin{equation}\label{eq:average-bound}
  \left|w(B_i)-\frac{w(E)}{k}\right|
  \le \left(1-\frac{1}{k}\right)\cdot\max_{e\in E}w_e
  \qquad\forall i\in[k].
\end{equation}
\end{theorem}

The bound in \eqref{eq:weighted-equitability} is best possible, even for a rank-1 uniform with one element of positive weight and all remaining elements of weight zero. Oki and Schwarcz~\cite{oki2025generalizing} independently proved a related weighted equitability result for matroids representable over a field of characteristic zero. They showed that the ground set can be partitioned into bases $B_1,\ldots,B_k$ satisfying
\[
w(B_1)\geq w(B_2)\geq\cdots\geq w(B_k)\geq w(B_1)-\max_{e\in B_1}w_e.
\]
This conclusion also follows, for arbitrary matroids, from the bundle-dependent guarantee of our truncation algorithm proved in Section~\ref{sec:fair}.

The proof first considers block matroids, that is, matroids whose ground set is the disjoint union of two bases. A block is a basis whose complement is also a basis. After scaling the element weights to lie in $[0,1]$, we show that the difference between two consecutive weights attained by blocks is at most one. The exchange conjecture of Greene and Magnanti~\cite{greene1975some}, in the special case later posed independently by Gabow~\cite{gabow1976decomposing}, would imply this statement by providing a sequence of symmetric exchanges between a block and its complement. We prove the statement without assuming this conjecture. Our argument combines linear programming duality with the exchange theorem of Akrami, Liu, Raj, and Végh~\cite{AkramiLiuRajVegh2026}. Applying the block result to a heaviest and a lightest basis repeatedly gives an equitable decomposition.

We also give a strongly polynomial-time algorithm. For a parameter $\tau$, truncate every element weight to $w_e^\tau=\min\{w_e,\tau\}$. Starting from an arbitrary partition into bases at $\tau=0$, the algorithm continuously increases $\tau$ while maintaining $|w^\tau(B_i)-w^\tau(B_j)|\leq \tau$ for every $i,j\in[k]$. Whenever an inequality becomes tight and would be violated by further increasing $\tau$, a single exchange is used for fixing the partition. A convex potential based on the numbers of active elements in the bases decreases after every repair step and whenever the active set changes. This bounds the total number of iterations.

Weighted equitability has two direct applications. The first is matroid-constrained load balancing on identical machines. The jobs form the ground set of a matroid, and the jobs assigned to every machine must be independent. This problem belongs to the family of optimal matroid partitioning problems studied by Kawase, Kimura, Makino, and Sumita \cite{kawase2021optimal}. By adding zero-processing-time dummy jobs, any feasible partition into independent sets can be completed to a partition into bases of a larger matroid. Applying weighted equitability gives a feasible schedule on $m$ machines with a makespan at most 
\[
\opt+\left(1-\frac{1}{m}\right)p_{\max},
\]
where $p_{\max}$ is the largest processing time. In particular, this yields a $(2-\frac{1}{m})$-approximation for the corresponding min--max matroid partitioning problem. 

The second application concerns fair division under a common matroid constraint. Suppose that $k$ agents have identical additive valuations and that a feasible allocation is a partition of the goods into $k$ bases. Akrami, Liu, Raj, and V\'egh \cite{AkramiLiuRajVegh2026} proved the existence of a matroid-constrained EF1 allocation for identical tri-valued additive valuations and left the case of general identical additive valuations open. Our weighted result resolves this question. Although weighted equitability alone only gives a bound in terms of the largest weight in the entire ground set, the truncation algorithm yields the bundle-dependent inequality
\[
    \min_{i\in[k]} w(B_i)\geq w(B_j)-\max_{e\in B_j}w_e
\]
for every final bundle $B_j$, which is equivalent to EF1 under identical additive valuations. 

We next consider discrepancy rounding of matroid bases. Let $x$ be a fractional point in the matroid base polytope $P(M)$, let $W$ be a nonnegative matrix with at most $\Delta$ nonzero entries in every column, and let $c$ be a cost vector. We give a polynomial-time algorithm that finds a basis $B$ of $M$ satisfying
\[  
c(B)\leq c^\top x\qquad\text{and}\qquad\left\|W\mathbf 1_B-Wx\right\|_\infty\leq 2\Delta\cdot\max_{i,e}W_{i,e}.
\]
This may be viewed as a matroidal analogue of the Beck--Fiala theorem~\cite{beck1981integer}.
It also gives a weighted analogue of the theorem of Király, Lau, and Singh~\cite{kiraly2008degree} on degree-bounded matroid bases, extending the setting from incidence matrices to arbitrary nonnegative matrices while retaining the same column-sparsity parameter. 
Specializing the matroid to a partition matroid gives an alternative proof of classical results for the generalized assignment problem (GAP) \cite{lenstra1990approximation,shmoys1993approximation}, up to a factor of two.


GAP has a prefix-constrained version. We are given a set of machines $I$ and ordered jobs $J=\{1,\ldots, n\}$; job $j$ has processing time $d_{ij}$ on machine $i$. The goal is to assign jobs to machines such that for every machine and every prefix of the jobs, the integral load is close to the corresponding fractional load.
Bansal, Rohwedder, and Svensson~\cite{bansal2022flow} conjectured the following:
\begin{conjecture}[Prefix-constrained GAP]\label{conj:prefix_ssuf}
    Let $G=(I\cup J,E)$ be a bipartite graph where $J=[n]$. For every fractional assignment $x\in [0,1]^{E}$ and $d\in \R_{\ge 0}^{E}$, there exists an assignment $y\in\{0,1\}^{E}$ such that for some universal constant $C>0$,
    \[
    \left|\sum_{j=1}^t d_{ij}y_{ij}-\sum_{j=1}^t d_{ij}x_{ij}\right|\leq C\cdot\max_{e\in E}d_e\qquad \forall i\in I,\ t\in [n].
    \]
\end{conjecture}
We formulate the following matroidal analogue. Given a matroid $M=([n],\cI)$, weights $w\in\R_{\geq 0}^n$, and a fractional basis $x\in P(M)$, does there exist a basis $B$ such that
\[
\left|\sum_{i=1}^t w_i\mathbf 1\{i\in B\}-\sum_{i=1}^t w_ix_i\right|\leq C\cdot\max_{e\in[n]}w_e
\]
for all $t\in[n]$ for a universal constant $C$? We show that this prefix-constrained basis conjecture implies the Prefix-constrained GAP Conjecture \ref{conj:prefix_ssuf} up to a factor of two. We also consider the cost-augmented version, in which the rounded basis is additionally required to satisfy $c(B)\leq c^\top x$. Using the face-preserving rounding framework of Swamy, Traub, Koch, and Zenklusen~\cite{swamy2026unsplittable}, we show that any bound for the original conjecture yields a cost-preserving bound with twice the error.

We prove a bound $C=O(\log n)$ for prefix-constrained bases. The proof uses iterated partial coloring, where we iteratively solve an LP relaxation that finds a partial coloring that fixes a constant fraction of the variables. We also prove two special cases. For binary weights, the cost-augmented version holds with $C=1$, using the total unimodularity of a system formed by two chains of constraints. For uniform matroids and arbitrary nonnegative weights, the conjecture also holds with $C=1$, via a direct rounding algorithm that processes the elements in order while preserving the weighted prefix sums.

Finally, we consider the restricted version of Conjecture \ref{conj:prefix_ssuf} where $d_{ij}=d_j$ for every $i$ incident to $j$, which is also known as the weighted carpooling problem~\cite{LiuReis2026Weighted}. In fact, a bound where $C=1$ was conjectured in \cite{LiuReis2026Weighted}:
\begin{conjecture}[Weighted carpool]\label{conj:weighted-carpool}
    Let $G=(I\cup J,E)$ be a bipartite graph where $J=[n]$. For every fractional assignment $x\in [0,1]^{E}$ and $d\in \R_{\ge 0}^{E}$, there exists an assignment $y\in\{0,1\}^{E}$ such that
    \[
    \left|\sum_{j=1}^t d_{j}y_{ij}-\sum_{j=1}^t d_{j}x_{ij}\right|\leq \max_{j\in [n]}d_j\qquad \forall i\in I,\ t\in [n].
    \]
\end{conjecture}

This conjecture is a special case of Morell and Skutella's conjecture on single-source unsplittable flows with arc-wise upper and lower bounds \cite{morell2022single} (see Appendix B of the same paper or \cite{LiuReis2026Weighted}). We give a counterexample to Conjecture \ref{conj:weighted-carpool}, which also disproves the Morell--Skutella conjecture; we additionally give a simpler counterexample directly to the latter.

\subsection{Further related work}

We briefly review related results and techniques relevant to the main results of the paper.

\paragraph{Gabow's conjecture.} Greene and Magnanti~\cite{greene1975some} conjectured that for any two bases $B$ and $D$ and every $X\subseteq B\setminus D$, there exists $Y\subseteq D\setminus B$ such that $X$ and $Y$ can be exchanged one element at a time while preserving both bases. The special case where $X=B\setminus D$ was later posed independently by Gabow~\cite{gabow1976decomposing}.

\begin{conjecture}[Gabow]\label{conj:Gabow}
    Let $B,D$ be bases of a matroid $M$. Then, there exist orderings $b_1,\ldots,b_r$ of $B$ and $d_1,\ldots,d_r$ of $D$ such that
\[
B_i=\left(B\setminus\{b_1,\ldots,b_i\}\right)\cup\{d_1,\ldots,d_i\}
\qquad\text{and}\qquad
D_i=\left(D\setminus\{d_1,\ldots,d_i\}\right)\cup\{b_1,\ldots,b_i\}
\]
are bases for every $i\in\{0,\ldots,r\}$. 
\end{conjecture}

The connection between Gabow's conjecture and matroid equitability was already observed in earlier work~\cite{berczi2023exchange}. Indeed, an affirmative answer to Gabow's conjecture implies the equitability theorem, and the same argument also yields the weighted Theorem~\ref{thm:weighted}. Gabow's conjecture is known to hold for several classes of matroids, including matching, strongly base-orderable, graphic, sparse paving, split, and regular matroids; see~\cite{berczi2023exchange,Berczi2024,oki2025generalizing} for further cases and references. Consequently, Theorem~\ref{thm:weighted} also follows from these results for the corresponding matroid classes. We prove Theorem~\ref{thm:weighted} for arbitrary matroids without assuming Gabow's conjecture. This is especially interesting in view of Larson's recent counterexample~\cite{larson2026counterexamples} to the closely related reconfiguration version of White's conjecture, according to which any two pairs of bases with the same multiset union can be connected by symmetric exchanges.

\paragraph{Rounding fractional bases.} Kir\'aly, Lau, and Singh~\cite{kiraly2008degree} studied degree-bounded matroid bases and gave a cost-preserving iterative-relaxation algorithm with additive violation $2\Delta-1$, where $\Delta$ is the maximum number of degree constraints containing any element. Bansal and Nagarajan~\cite{bansal2017approximation} generalized the discrepancy-rounding framework by Lovett and Meka~\cite{lovett2015constructive} to matroid polytopes, proving an additive violation $\min\{O(\sqrt{\Delta}\log{n}), O(\sqrt{n\log(m/n)})\}$. 
This improves upon the bound $O(\sqrt{n\log m})$ which is a consequence of the concentration property of randomized swap rounding for matroid polytopes by Chekuri, Vondr\'ak, and Zenklusen~\cite{chekuri2010dependent}. Those results are interesting in the regime where the column sparsity $\Delta$ is large. Bansal~\cite{bansal2019generalization} further prove an additive guarantee $O(\Delta)$ for degree bounded bases with additional concentration properties.

\paragraph{Budgeted matroid intersection.} Budgeted matroid intersection is also related to weighted equitability in the case of two bases. Let $M$ be a block matroid and set $w_{\max}\coloneqq\max_{e\in E}w_e$. Since the blocks of $M$ are precisely the common bases of $M$ and its dual $M^*$, the case $k=2$ of weighted equitability is equivalent to finding a common basis $B$ such that $w(E)/2-w_{\max}/2\leq w(B)\leq w(E)/2$. Indeed, given an equitable decomposition into two bases, one may take the lighter of the two. Conversely, any common basis satisfying these inequalities, together with its complement, forms an equitable decomposition. Budgeted matroid intersection instead asks for a maximum-weight common independent set of two matroids subject to a budget constraint. Taking both the weight and the cost function to be $w$, with budget $w(E)/2$, gives a relaxation of the common basis formulation above, in which the solution is only required to be a common independent set. 
Berger, Bonifaci, Grandoni, and Sch\"afer~\cite{berger2011budgeted} gave a PTAS for this problem. Chekuri, Vondr\'ak, and Zenklusen~\cite{chekuri2011multi} further gave a PTAS for budgeted matroid intersection with $k$ budget constraints for a fixed $k$. More recently, Doron-Arad, Kulik, and Shachnai~\cite{doronArad2026eptas} gave an EPTAS. 

\paragraph{Prefix Beck--Fiala.} Given $n$ vectors $v_1,\ldots,v_n$ such that $\onenorm{v_j}\le 1$ for every $j\in [n]$, it is conjectured that there exists a signing $\varepsilon\in\{-1,1\}^n$ such that the prefix sum $
\inftynorm{\sum_{j=1}^t \varepsilon_jv_j}\le C$ for every $t\in [n]$ for some constant $C$. This is known as the \emph{prefix Beck--Fiala} conjecture, since the non-prefix version, which concerns only $t=n$, was proved by Beck and Fiala \cite{beck1981integer} with $C=2$. The best known bound for prefix Beck-Fiala is $C=O(\sqrt{\log n})$ due to Banaszczyk \cite{banaszczyk2012series}. Very recently, Aden-Ali \cite{aden2026optimal} gives a linear-time online algorithm achieving the same bound.

Bansal, Rohwedder and Svensson \cite{bansal2022flow} proved that Conjecture \ref{conj:prefix_ssuf} is equivalent to the special case of \emph{$2$-sparse} prefix Beck--Fiala, where every vector $v_j$ has at most $2$ nonzero entries, up to a constant factor. 
The main motivation for studying Conjecture \ref{conj:prefix_ssuf} is that it implies a constant-factor approximation for \emph{max flow-time scheduling} for unrelated machines: each job has a release date $r_j$ and a completion time $c_j$; the goal is to find an assignment that minimizes $\max_j c_j-r_j$. The current best approximation ratio is an $O(\sqrt{\log n})$-approximation algorithm combining results of \cite{bansal2022flow} and \cite{aden2026optimal}.

\paragraph{Single-source unsplittable flows.} 

Let $D=(V,A)$ be a directed acyclic graph with source $s \in V$, sink terminals $t_1, \dots, t_n \in V$ and associated demands $d \in \R^n_{> 0}$. A flow $x \in \R^{A}_{\ge 0}$ satisfies the demands if $x(\delta^-(t_j))=d_j$ for every $j\in [n]$, and $x(\delta^-(v))=x(\delta^+(v))$ for every $v\in V\setminus \{s,t_1,\dots,t_n\}$. We say a flow $y \in \R^{A}_{\ge 0}$ is \emph{unsplittable} if for each $j \in [n]$ there is a single path $P_j$ from $s$ to $t_j$ that carries $d_j$ units of flow, so that $y_a=\sum_{j: a\in P_j} d_j$ for every $a\in A$. Morell and Skutella \cite{morell2022single} conjectured the following:

\begin{conjecture}[Morell, Skutella]\label{conj:ssuf}
   For every flow $x \in \R^{A}_{\ge 0}$ satisfying the demands, there is an unsplittable flow $y \in \R^{A}_{\ge 0}$ such that $|x_a-y_a|\le \max_{j\in [n]} d_j$ for all $a \in A$.
\end{conjecture}
Dinitz, Garg and Goemans~\cite{DinitzGargGoemans1999Combinatorica} initiated the line of work on single-source unsplittable flows, and proved the existence of an unsplittable flow satisfying the upper bound $y_a\le x_a+\max_{j \in [n]} d_j$ for all $a \in A$. Goemans also conjectured that given costs $c\in \R_{\ge 0}^A$, there exists an unsplittable flow satisfying the upper bound and $c^\top y\le c^\top x$, which was disproved very recently by Rybin, assisted by AI \cite{Rybin2026DGGCounterexample}.
Conjecture~\ref{conj:ssuf} was proved for the special cases where the demands have the property that one divides another, i.e., $d_1\mid d_2\mid \ldots \mid d_n$ \cite{morell2022single} and for special digraphs such as acyclic planar digraphs~\cite{TraubVargasKochZenklusen2024SODA}; see also~\cite{AlmoghrabiSkutellaWarode2025IPCO} for stronger guarantees for series-parallel digraphs. Liu and Reis \cite{LiuReis2026Weighted} proved a special case of the weighted carpool Conjecture \ref{conj:weighted-carpool}, where the bipartite graph $G$ is complete. They also conjectured that, in this case, the more general prefix-constrained GAP Conjecture \ref{conj:prefix_ssuf} holds for $C=1$. This was again disproved very recently by Muffatti \cite{Muffatti2026WeightedRoundingCounterexamples} with the help of AI. See Figure~\ref{fig:conjecture-implication-map} for the relationship of conjectures and theorems considered in this paper.

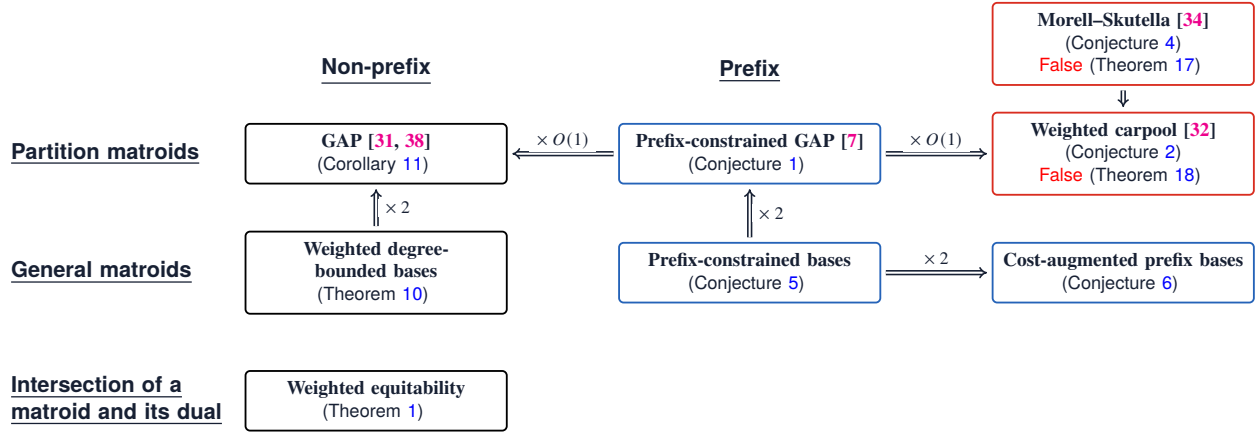
\begin{figure*}[t]
\centering
\resizebox{\textwidth}{!}{%
\begin{tikzpicture}[x=1cm,y=1cm]

\tikzset{
  mapstatement/.style={
    rectangle,
    rounded corners=1mm,
    fill=white,
    line width=1.15pt,
    text=mapink,
    align=center,
    text width=5.5cm,
    inner xsep=1mm,
    inner ysep=2.2mm,
    font=\large
  },
  mapopenbox/.style={
    mapstatement,
    draw=mapopen
  },
  mapfalsebox/.style={
    mapstatement,
    draw=mapfalse
  },
  maptheorembox/.style={
    mapstatement,
    draw=maptheorem
  },
  mapimplies/.style={
    draw=mapink,
    line width=1pt,
    double=white,
    double distance=1.3pt,
    -{Implies[length=5mm,width=5mm]},
    shorten >=1mm,
    shorten <=1mm
  },
  mapfactor/.style={
    fill=white,
    inner xsep=1.5mm,
    inner ysep=0.6mm,
    text=mapink,
    font=\sffamily\bfseries
  },
  mapcolumn/.style={
    font=\sffamily\Large\bfseries,
    text=mapink,
    align=center
  },
  maprow/.style={
    font=\sffamily\Large\bfseries,
    text=mapink,
    align=left,
    text width=4.5cm
  }
}

\newcommand{\conjmeta}[1]{%
  {\sffamily#1}%
}


\node[mapcolumn] at (-6.60,4.45)
  {\underline{Non-prefix}};

\node[mapcolumn] at (1.55,4.45)
  {\underline{Prefix}};

\node[maprow,anchor=west] at (-14.70,2.65)
  {\underline{Partition matroids}};

\node[maprow,anchor=west] at (-14.70,0.05)
  {\underline{General matroids}};

\node[maprow,anchor=west] at (-14.70,-2.75) {
  \underline{Intersection of a}\\[-0.5mm]
  \underline{matroid and its dual}
};

\node[maptheorembox] (lst) at (-6.60,2.65) {
  {\bfseries GAP \cite{lenstra1990approximation,shmoys1993approximation}}
  \par
  \conjmeta{
    (Corollary \ref{cor:LST-2})
  }
};

\node[mapopenbox] (c1) at (1.55,2.65) {
  {\bfseries Prefix-constrained GAP \cite{bansal2022flow}}\par
  \conjmeta{
    (Conjecture \ref{conj:prefix_ssuf})
  }
};

\node[mapfalsebox] (c2) at (9.70,2.65) {
  {\bfseries Weighted carpool \cite{LiuReis2026Weighted}}
  \par
  \conjmeta{
    (Conjecture \ref{conj:weighted-carpool})\\
    {\color{red}False} (Theorem \ref{thm:counterexample-carpool})
  }
};

\node[mapfalsebox] (c4) at (9.70,5.05) {
  {\bfseries Morell--Skutella \cite{morell2022single}}\par
  \conjmeta{
    (Conjecture \ref{conj:ssuf})\\
    {\color{red}False} (Theorem \ref{thm:counterexample-MS})
  }
};

\node[maptheorembox] (degree) at (-6.60,0.05) {
  {\bfseries Weighted degree-bounded bases}\par
  \conjmeta{
    (Theorem \ref{thm:weighted_degree_bounded_basis})
  }
};

\node[mapopenbox] (c5) at (1.55,0.05) {
  {\bfseries Prefix-constrained bases}\par
  \conjmeta{
    (Conjecture \ref{conj:matroid_chain})
  }
};

\node[mapopenbox] (c6) at (9.70,0.05) {
  {\bfseries Cost-augmented prefix bases}\par
  \conjmeta{
    (Conjecture \ref{conj:matroid_chain-cost})
  }
};

\node[maptheorembox] (weighted) at (-6.60,-2.75) {
  {\bfseries Weighted equitability}\par
  \conjmeta{
    (Theorem \ref{thm:weighted})
  }
};


\draw[mapimplies]
  (c1.west) --
  node[mapfactor,above=1pt] {$\times~O(1)$}
  (lst.east);

\draw[mapimplies]
  (c1.east) --
  node[mapfactor,above=1pt] {$\times~O(1)$}
  (c2.west);

\draw[mapimplies]
  (c4.south) -- (c2.north);


\draw[mapimplies]
  (degree.north) --
  node[mapfactor,right=1pt] {$\times~2$}
  (lst.south);



\draw[mapimplies]
  (c5.north) --
  node[mapfactor,right=1pt] {$\times~2$}
  (c1.south);

\draw[mapimplies]
  (c5.east) --
  node[mapfactor,above=3pt] {$\times~2$}
  (c6.west);

\end{tikzpicture}%
}
\caption{Relationship between the theorems and conjectures. The ones in the black boxes are proved; the ones in the red boxes are disproved; the ones in the blue boxes are open. The arrow $\Rightarrow$ means implication up to a factor written on the arrow.}
\label{fig:conjecture-implication-map}
\end{figure*}

\subsection{Notation and preliminaries}
\label{sec:prelim}

For a positive integer $k$, let $[k]\coloneqq\{1,\ldots,k\}$. For a vector $x\in\mathbb{R}^E$ and a set $A\subseteq E$, we write $x(A)\coloneqq\sum_{e\in A}x_e$, and denote the \emph{characteristic vector} of $A$ by $\ind_A$. For a set $A$ and an element $e$, we write $A-e\coloneqq A\setminus\{e\}$ and $A+e\coloneqq A\cup\{e\}$. Denote by $A\triangle B$ for symmetric difference of sets $A$ and $B$. For nonnegative weights, we adopt the convention that $\max_{e\in\emptyset}w_e=0$. In algorithmic definitions, the minimum of an empty set is defined to be $+\infty$.

Let $M=(E,\cI)$ be a matroid. We denote its rank function by $r_M$ and its family of bases by $\cB(M)$. For $A\subseteq E$, we denote the \emph{restriction} of $M$ to $A$ by $M|_A$, and \emph{deletion} and \emph{contraction} by $M\backslash A$ and $M/A$, respectively. The \emph{dual} matroid is denoted by $M^*$. A matroid $M=(E,\cI)$ is a \emph{block matroid} if its ground set can be partitioned into two bases. A set $B\subseteq E$ is a \emph{block} if both $B$ and $E\setminus B$ are bases of $M$. Equivalently, $B$ is a common basis of $M$ and its dual $M^*$.

For matroids on disjoint ground sets, their \emph{direct sum} is denoted by $M_1\oplus M_2$. The \emph{uniform matroid} of rank $r$ on $n$ elements is denoted by $U_{r,n}$. For a nonnegative integer $k$, the rank-$k$ \emph{truncation} of $M$ is denoted by $M_k$ and has rank function
\[
r_{M_k}(A)=\min\{k,r_M(A)\}
\qquad\forall A\subseteq E.
\]

We work in the independence-oracle model. Our existence results allow arbitrary real weights. In algorithmic statements, the weights are assumed to be rational. Since the algorithms and their guarantees are invariant under positive scaling of the weights, we may clear denominators and describe the algorithms for integral weights. Running times count arithmetic operations, comparisons, and independence-oracle calls.

\medskip

\paragraph{Organization.} 

Section~\ref{sec:weighted} proves weighted matroid equitability, and Section~\ref{sec:algorithm} gives the strongly polynomial-time algorithm. Section~\ref{sec:applications} presents the scheduling and fair-division applications. Section~\ref{sec:degree_bounded} proves the weighted degree-bounded basis theorem and discusses its connection to generalized assignment. Section~\ref{sec:prefix} introduces prefix-constrained bases, proves an $O(\log n)$ bound, and establishes the special cases described above. Section~\ref{sec:unsplittable} gives the counterexamples for weighted carpooling and single-source unsplittable flows.

\section{Weighted equitability}
\label{sec:weighted}



We prove the weighted equitability Theorem~\ref{thm:weighted} in this section through a more precise statement for block matroids. We show that if the element weights lie in $[0,1]$, the difference between two consecutive block weights is at most $1$. Then, applying this statement to a heaviest and a lightest basis yields the theorem.
\begin{theorem}\label{thm:scalar-gap}
Let $M=(E,\cI)$ be a block matroid, let $\cF$ be its family of blocks, and let $w\in[0,1]^E$. If $q_1<\cdots<q_m$
are the distinct values in $\{w(B)\colon B\in\cF\}$, then
\[
  q_{s+1}-q_s\le1
  \qquad\forall s\in[m-1].
\]
\end{theorem}
Note that Gabow's Conjecture \ref{conj:Gabow} would imply Theorem \ref{thm:scalar-gap}. Indeed, let $q_s<q_{s+1}$ be consecutive block weights. If $q_s<w(E)/2$, take a block $B$ of weight $q_s$. Its complement has weight $w(E)-q_s>q_s$. Conjecture \ref{conj:Gabow} gives an exchange sequence from $B$ to $E\setminus B$, each one being a block. Consider the first block in the sequence whose weight is greater than $q_s$. Its predecessor has weight at most $q_s$, while its weight is at least $q_{s+1}$. Since the two blocks differ by a single exchange, it follows that $q_{s+1}-q_s\le 1$. If $q_s\ge w(E)/2$, the same argument applies in the opposite direction, starting from a block of weight $q_{s+1}$ and ending at its complement.

To prove Theorem~\ref{thm:scalar-gap}, we crucially use the following exchange theorem of Akrami, Liu, Raj, and V\'egh
\cite[Theorem~1.4]{AkramiLiuRajVegh2026}. Let $B$ and $D$ be two disjoint bases in the matroid $M=(E,\cI)$. For $S\subseteq B\cup D$ and $t\in (B\cup D)\setminus S$ a set $X\subseteq B\cup D$ is \emph{$(t,S)$-exchangeable} for $B$ and $D$ if $$t\in X\subseteq S+t$$ and both $B\triangle X$ and $D\triangle X$ are bases.

\begin{theorem}[Akrami, Liu, Raj, Végh]\label{thm:unweighted-exchange}
    Let $B,D$ be disjoint bases of a matroid $M=(E,\cI)$, and let $S\subseteq E$. If $|B\cap S|>|D\cap S|$, then there exist $t\in D\setminus S$ and a $(t,S)$-exchangeable set for $B$ and $D$. Moreover, such a set can be found in polynomial time.
\end{theorem}

\begin{proof}[Proof of Theorem~\ref{thm:scalar-gap}]
     Suppose for the sake of contradiction that $q_{s+1}-q_s>1$ for some $s\in[m-1]$. Then the blocks can be partitioned into $\cF=\cF^+\sqcup\cF^-$, so that
    \begin{equation}\label{eq:initial-gap}
      \min_{B\in\cF^+}w(B)-\max_{D\in\cF^-}w(D)>1.
    \end{equation}
    For this fixed partition, consider the linear program
    \begin{equation*}\label{eq:primal-LP}
    \begin{aligned}
      \max\quad &\gamma\\
      \text{s.t.}\quad &u(B)-u(D)\ge\gamma
          &&\forall B\in\cF^+,\ D\in\cF^-,\\
      &0\le u_e\le1
          &&\forall e\in E.
    \end{aligned}
    \end{equation*}
    Its dual can be written as
    \begin{equation*}\label{eq:dual-LP}
    \begin{aligned}
      \min\quad &\sum_{e\in E}x_e\\
      \text{s.t.}\quad
      &\sum_{B\in\cF^+,\,D\in\cF^-}
           \lambda_{B,D}(\ind_B-\ind_D)\le x,\\
      &\sum_{B\in\cF^+,\,D\in\cF^-}\lambda_{B,D}=1,\\
      &x\ge0,\quad\lambda\ge0.
    \end{aligned}
    \end{equation*}
    Let $(u,\gamma)$ be an optimal primal solution, and set
    \[
      b\coloneqq \min_{B\in\cF^+}u(B),
      \qquad
      d\coloneqq \max_{D\in\cF^-}u(D).
    \]
    Then $\gamma=b-d>1$, since the original weight vector $w$ gives a feasible solution of value greater than one by \eqref{eq:initial-gap}.
    
    Let $(\lambda,x)$ be an optimal dual solution and define $y\coloneqq \sum_{B\in\cF^+,\,D\in\cF^-} \lambda_{B,D}(\ind_B-\ind_D)$. Notice that in the optimal solution, $x_e=\max\{0,y_e\}$. By complementary slackness,
    \begin{equation}\label{eq:signs}
      \begin{cases}
      y_e\le0& \text{if $u_e=0$},\\
      y_e=0& \text{if $0<u_e<1$},\\
      y_e\ge0& \text{if $u_e=1$}.
      \end{cases}
    \end{equation}
    Define
    \[
      S\coloneqq \{e\in E\colon u_e=1\}.
    \]
    Further, we define
    \[
    \begin{aligned}
      \cH^+&\coloneqq \{B\in\cF^+\colon \lambda_{B,D}>0
                      \text{ for some }D\in\cF^-\},\\
      \cH^-&\coloneqq \{D\in\cF^-\colon\lambda_{B,D}>0
                      \text{ for some }B\in\cF^+\}.
    \end{aligned}
    \]
    One can think of $\lambda$ as a distribution over $\cF^+\times \cF^-$, and $\cH^+$ and $\cH^-$ are the supports of the two marginals of $\lambda$ over $\cF^+$ and $\cF^-$, respectively. Then, we have
    \begin{equation}\label{eq:LP-duality}
    \begin{aligned}
        \gamma
        &=\sum_{B\in\cH^+,\,D\in\cH^-}
            \lambda_{B,D}\left(u(B)-u(D)\right)\\
        &=u^\top y\\
        &=y(S)\\
        &=\sum_{B\in\cH^+,\,D\in\cH^-}
            \lambda_{B,D}\left(|B\cap S|-|D\cap S|\right),
    \end{aligned}
    \end{equation}
    where the second equality uses complementary slackness \eqref{eq:signs}.

    We claim that
    \begin{equation}\label{eq:Hplus-count}
      |B\cap S|\le\frac{|S|}{2}
      \qquad\forall B\in\cH^+.
    \end{equation}
    Indeed, it follows from complementary slackness that for every $B\in\cH^+$, there exists $D\in\cF^-$ with $\lambda_{B,D}>0$ such that $u(B)-u(D)=\gamma$. This implies $u(B)=\gamma+u(D)\le\gamma+d=b$, which means that every block in $\cH^+$ attains minimum weight $b$. Assume for the sake of contradiction that $|B\cap S|>|S|/2$. According to Theorem~\ref{thm:unweighted-exchange}, there exists $e\in(E\setminus B)\setminus S$ and an $(e,S)$-exchangeable set $X$ for $B$ and $E\setminus B$. Then $B'\coloneqq B\mathbin\triangle X$ is a block. Since $B'$ and $B$ have the same cardinality and every element of $X-e$ has $u$-weight one,
    \[
      u(B')=u(B)+u_e-1\in\left[u(B)-1,u(B)\right)=\left[b-1,b\right).
    \]
    This is a contradiction, because we know every block has either weight at least $b$ if it is in $\cF^+$, or weight at most $d<b-1$ if it is in $\cF^-$. This proves \eqref{eq:Hplus-count}.
    
    The symmetric argument applies to show that
    \begin{equation}\label{eq:Hminus-count}
      |D\cap S|\ge\frac{|S|}{2}
      \qquad\forall D\in\cH^-.
    \end{equation}
    Plugging \eqref{eq:Hplus-count} and \eqref{eq:Hminus-count} into \eqref{eq:LP-duality}, we obtain $\gamma\le0$, which is a contradiction to $\gamma>1$.
\end{proof}

We now use Theorem~\ref{thm:scalar-gap} to prove Theorem~\ref{thm:weighted}. Starting from an arbitrary partition into bases, we repeatedly rebalance a heaviest and a lightest basis by applying Theorem~\ref{thm:scalar-gap} to their union.

\begin{proof}[Proof of Theorem~\ref{thm:weighted}]
    If $\max_{e\in E}w_e=0$, then the statement is immediate. By scaling we may therefore assume without loss of generality that $\max_{e\in E} w_e=1$. Start from an arbitrary partition into $k$ bases $E=B_1\sqcup\cdots\sqcup B_k$. If there exist $B_i,B_j$ such that $w(B_i)-w(B_j)>1$, pick $B_i$ with the maximum weight and $B_j$ with the minimum weight among the $k$ bases. Consider the block matroid $M'=(E',\cI')=M|_{B_i\sqcup B_j}$, which is the restriction of $M$ to $B_i\sqcup B_j$. We pick a partition $E'=B_i'\sqcup B_j'$ of $M'$ into bases that minimizes
    \[
        |w(B_i')-w(B_j')|
        =\left|w(B_i')-\left(w(E')-w(B_i')\right)\right|
        =2\left|w(B_i')-\frac{w(E')}{2}\right|.
    \]
    
    We claim that $|w(B_i')-w(B_j')|\le 1$. Indeed, picking $B_i'$ to be a maximum-weight block gives $w(B_i')-w(B_j')\ge 0$, while picking $B_i'$ to be a minimum-weight block gives $w(B_i')-w(B_j')\le 0$. If some block has weight $w(E')/2$, then the claim follows immediately. Otherwise, $0$ lies in the interval
    \[
    I_s\coloneqq
    \left[q_s-\frac{w(E')}{2},q_{s+1}-\frac{w(E')}{2}\right]
    \]
    for some consecutive values $q_s,q_{s+1}\in\{w(B)\colon B\text{ is a block of }M'\}$. By Theorem~\ref{thm:scalar-gap}, $0<q_{s+1}-q_s\le 1$. Hence the endpoint of $I_s$ that is closer to $0$ satisfies $\left|q_t-w(E')/2\right|\le 1/2$ for some $t\in\{s,s+1\}$. Letting $B_i'$ be a block with weight $q_t$ gives $|w(B_i')-w(B_j')|\le 1$.
    
    After we replace the partition $B_i\sqcup B_j$ with $B_i'\sqcup B_j'$, either $\max_t w(B_t)$ decreases or $\max_t w(B_t)$ remains the same but the number of bases attaining the maximum weight decreases. Since there are only finitely many partitions of $E$ into bases, the process terminates. At termination, we have a partition such that $|w(B_i)-w(B_j)|\le 1$ for all $i,j\in[k]$. It follows that
    \[
    \left|w(B_i)-\frac{w(E)}{k}\right|
    =\left|\frac1k\sum_{j=1}^k\left(w(B_i)-w(B_j)\right)\right|
    \le\frac{k-1}{k}. \qedhere
    \]
\end{proof}

For a matroid $M=(E,\mathcal I)$ whose rank function is $r:2^E\rightarrow \Z_{\ge 0}$, its \emph{base polytope} is defined as
\[
    P(M)\coloneqq \{x\in\mathbb R_{\ge 0}^E\colon x(S)\le r(S)\ \forall S\subseteq E,\quad x(E)=r(E)\}.
\]
Edmonds' matroid partition theorem implies that a matroid $M$ can be partitioned into $k$ bases if and only if $\frac{1}{k}\cdot\ind\in P(M)$ \cite{edmonds1965minimum}. In particular, if $M$ is a block matroid, then $\frac{1}{2}\cdot\ind\in P(M)\cap P(M^*)$. The following is a direct corollary of Theorem~\ref{thm:scalar-gap}, which generalizes Theorem~\ref{thm:weighted} for the case where $k=2$.

\begin{corollary}\label{cor:equitability-block-matroid}
    Let $M=(E,\mathcal{I})$ be a block matroid, and let $w\in \R_{\ge 0}^E$. For every $x\in P(M)\cap P(M^*)$, there exists a block $B$ of $M$ such that $|w(B)-w^\top x|\le \frac{1}{2}\cdot\max_{e\in E}w_e$.
\end{corollary}
\begin{proof}
    Scale the weights so that $\max_{e\in E} w_e=1$. Since $x\in P(M)\cap P(M^*)$, by the integrality of matroid intersection polytope~\cite{edmonds2003submodular}, the maximum weighted common base $D$ of $M$ and $M^*$ satisfies $w(D)\ge w^\top x$. In other words, $D$ is a block of $M$ of maximum weight. Then, $E\setminus D$ is a block of minimum weight. Therefore, $w(D)\ge w^\top x\ge w(E\setminus D)$. Let $q_1<\cdots<q_m$ be distinct weights of blocks and suppose $q_s\le w^\top x<q_{s+1}$. It follows from Theorem~\ref{thm:scalar-gap} that $q_{s+1}-q_s\le 1$. Let $B$ be a block whose weight is either $q_s$ or $q_{s+1}$ whichever is closer to $w^\top x$. Then, $|w(B)-w^\top x|\le \frac{q_{s+1}-q_s}{2}\le \frac{1}{2}$.
\end{proof}
One can ask whether Corollary \ref{cor:equitability-block-matroid} generalizes to the intersection of two general matroids $M_1,M_2$. Namely, given a fractional $x\in P(M_1)\cap P(M_2)$, can we find a common basis $B$ of $M_1$ and $M_2$ such that $|w(B)-w^\top x|$ is small? It turns out that  the assumption that one matroid is the dual of the other is crucial. Consider the example of bipartite matching which is the intersection of two partition matroids. Let the entire graph be an even cycle of length $n$, and let $w$ be the indicator vector of the even edges in the cycle. Then, $x=\frac{1}{2}\cdot \ind$ is in the matroid intersection polytope, and $w^\top x=\frac{n}{4}$. Then, the common basis $B$ is either the matching consisting of even edges whose weight is $\frac{n}{2}$ or the one consisting of odd edges whose weight is $0$, both having $|w(B)-w^\top x|=\frac{n}{4}$.
\section{Strongly polynomial-time algorithm}
\label{sec:algorithm}

We now give a strongly polynomial-time algorithm that finds an equitable partition. The algorithm uses the algorithmic version of Theorem~\ref{thm:unweighted-exchange} as a subroutine and makes only a strongly polynomial number of calls to it.

\begin{theorem}
    For rational weights, an equitable partition satisfying \eqref{eq:weighted-equitability} can be found in strongly polynomial time.
\end{theorem}
\begin{proof}
For $0\le \tau\le \max_{e\in E}w_e$, define the truncated weight at $\tau$ as $w^{\tau}_e=\min\{w_e,\tau\}$. We give an algorithm so that for every $\tau$, we maintain a partition $E=B_1\sqcup \cdots \sqcup B_k$ into bases satisfying
\begin{equation}\label{eq:equitable_leve_tau}
    |w^\tau(B_i)-w^\tau(B_j)|\le \tau
\end{equation}
for all $i,j\in[k]$. For a fixed $\tau$, define the active set as
\[
S\coloneqq \{e\in E\colon w_e>\tau\}.
\]
For each $B_i$, let \[a_i\coloneqq w(B_i\setminus S), \qquad b_i\coloneqq |B_i\cap S|.\] 
Then,
\[
w^\tau(B_i)=a_i+b_i\tau.
\]
When $\tau$ is between two consecutive weights in $\{w_e\colon e\in E\}$, $S$ stays the same and the truncated weight of a basis $B_i$ is a linear function in $\tau$. Condition \eqref{eq:equitable_leve_tau} is equivalent to
\[
-\tau\le a_i-a_j+(b_i-b_j)\tau\le \tau\qquad\forall i,j\in [k], 
\]
that is,
\begin{subequations}\label{eq:linear}
\begin{align}
    (b_i-b_j-1)\tau&\le a_j-a_i \qquad\forall i,j\in [k]\label{eq:linear1},\\
    (b_i-b_j+1)\tau&\ge a_j-a_i \qquad\forall i,j\in [k]\label{eq:linear2}.
\end{align}
\end{subequations}
Initialize with $\tau=0$ and an arbitrary partition $E=B_1\sqcup \cdots \sqcup B_k$, which can be computed in strongly polynomial time using Edmonds' matroid partition algorithm \cite{edmonds2009matroid}. Then, \eqref{eq:equitable_leve_tau} is satisfied trivially since $w^0=0$. We increase $\tau$ continuously and update the partition if necessary. Suppose \eqref{eq:linear} is satisfied for the current $\tau$. As we increase $\tau$ by a sufficiently tiny amount $\varepsilon>0$, \eqref{eq:linear1} is violated for some $i,j\in [k]$ if and only if
\begin{equation}\label{eq:breakpoint1}
    b_i-b_j-1>0\qquad \text{and}\qquad (b_i-b_j-1)\tau=a_j-a_i.
\end{equation}
Similarly, \eqref{eq:linear2} is violated for $\tau+\varepsilon$ for some $i,j\in [k]$ if and only if
\begin{equation}\label{eq:breakpoint2}
    b_i-b_j+1<0\qquad \text{and}\qquad (b_i-b_j+1)\tau=a_j-a_i.
\end{equation}
Notice that $(i,j)$ satisfies \eqref{eq:breakpoint1} if and only if $(j,i)$ satisfies \eqref{eq:breakpoint2}. Thus it suffices to consider the case where \eqref{eq:breakpoint1} occurs.

For such $\tau=\frac{a_j-a_i}{b_i-b_j-1}$, we need to repair the partition. Since $|B_i\cap S|-|B_j\cap S|=b_i-b_j>1$, it follows from Theorem~\ref{thm:unweighted-exchange} that there exists $t\in B_j\setminus S$ and $X$ with $t\in X\subseteq S+t$ so that $B_i'\coloneqq B_i\triangle X$ and $B_j'\coloneqq B_j\triangle X$ are both bases. The new partition decreases $|B_i\cap S|-|B_j\cap S|$ by $2$. We repeatedly apply Theorem~\ref{thm:unweighted-exchange} and recompute $a_i,b_i, a_j,b_j$ until \eqref{eq:breakpoint1} is no longer satisfied. This completes the repair steps. If no repair step is needed, we increase $\tau$. The algorithm is presented as Algorithm~\ref{alg:weighted-equitability}.

\begin{algorithm}[htbp]
\caption{\textsc{Weighted Equitability}}
\label{alg:weighted-equitability}
\begin{algorithmic}[1]
\Require A matroid \(M=(E,\cI)\) whose ground set can be partitioned into $k$ bases, and weights \(w\in\mathbb{Q}_{\geq0}^E\).
\Ensure An equitable partition of $E$ into $k$ bases.

\State Initialize with \(\tau\gets0\) and an arbitrary partition $E=B_1\sqcup\cdots\sqcup B_k$.
\State \(W\gets\max_{e\in E}w_e\)

\While{\(\tau<W\)}
    \State \(S\gets\{e\in E:w_e>\tau\}\)
    \State \(a_i\gets w(B_i\setminus S),\ b_i\gets |B_i\cap S|\quad \forall i\in [k]\)

    \If{\(\exists i,j\in[k]\) such that $b_i-b_j>1$ and $(b_i-b_j-1)\tau=a_j-a_i$}
        \State Compute a $(t,S)$-exchangeable set $X$ for $B_i$ and $B_j$ for some $t\in B_j\setminus S$ using Theorem~\ref{thm:unweighted-exchange}.
        \State \(B_i\gets B_i\triangle X\), \(B_j\gets B_j\triangle X\)
        \State \textbf{continue}
        \Comment{Recompute \(a,b\) at the same \(\tau\)}
    \EndIf

    \State $\alpha\gets \min\left\{\frac{a_j-a_i}{b_i-b_j-1}\colon i,j\in [k] \text{ such that } b_i-b_j-1>0 \text{ and }\frac{a_j-a_i}{b_i-b_j-1}>\tau\right\}$
    \State $\beta\gets\min\{w_e:e\in E \text{ such that }w_e>\tau\}$
    \State $\tau\gets \min\{\alpha,\beta\}$.
    \Comment{Advance to the next breakpoint}
\EndWhile

\State \Return \(B_1,\ldots,B_k\)
\end{algorithmic}
\end{algorithm}

First, we claim that \eqref{eq:equitable_leve_tau} is maintained throughout the algorithm. This will prove the theorem by taking $\tau=\max_{e\in E}w_e$. Indeed, when we do a repair step, 
\begin{equation}\label{eq:diff-before}
    w^\tau(B_i)-w^\tau(B_j)=a_i-a_j+(b_i-b_j)\tau=\tau.
\end{equation}
After repairing,
\begin{equation}\label{eq:diff-after}
w^\tau(B'_i)-w^\tau(B'_j)=w^\tau(B_i)-w^\tau(B_j)-2(\tau-w_t)=\tau-2(\tau-w_t)\in [-\tau,\tau].
\end{equation}
Here we used the fact that $0\le w_t\le\tau$, since $t\notin S$. Moreover, before the repair, \eqref{eq:equitable_leve_tau} and \eqref{eq:diff-before} imply that
\[
    w^\tau(B_j)\le w^\tau(B_h)\le w^\tau(B_i)
\]
for all $h\in[k]$. The repair decreases $w^\tau(B_i)$ by $\tau-w_t$ and increases $w^\tau(B_j)$ by the same amount. Hence both $w^\tau(B_i')$ and $w^\tau(B_j')$ remain in the interval $[w^\tau(B_j),w^\tau(B_i)]$, while the weights of all other bases remain unchanged. Thus, after repairing, \eqref{eq:equitable_leve_tau} is satisfied. Moreover, after all repair steps at the current value of $\tau$ have been completed, every tight inequality in \eqref{eq:linear1} satisfies $b_i-b_j\le1$, and every tight inequality in \eqref{eq:linear2} satisfies $b_i-b_j\ge-1$. Therefore, no inequality in \eqref{eq:linear} is violated when we increase $\tau$, as the left-hand sides change in directions that preserve the inequalities.

We are left to prove that the algorithm terminates. We define the potential to be $\Phi=\sum_{i=1}^k\left |B_i\cap S\right|^2$. When we do a repair step for bases $B_i, B_j$ where $b_i-b_j-1>0$, 
\[
\begin{aligned}
    \Phi'-\Phi&=|B'_i\cap S|^2+|B'_j\cap S|^2-|B_i\cap S|^2-|B_j\cap S|^2\\
    &=(b_i-1)^2+(b_j+1)^2-b_i^2-b_j^2\\
    &=-2(b_i-b_j-1)\\
    &<0.
\end{aligned}
\]
Thus, the potential decreases by at least $1$ at every repair step. When the active set $S$ changes as $\tau$ increases, at least one element leaves $S$, and hence $\Phi$ also decreases by at least $1$. Since $\Phi\le kr^2$, the total number of repair steps and times the active set changes is at most $kr^2$. Every increase of $\tau$ reaches either a breakpoint at which a repair step is needed or the next element weight, at which the active set changes. Therefore, the algorithm performs $O(kr^2)=O(|E|r)$ iterations. 
Each repair step can be carried out in strongly polynomial time using the algorithm from Theorem~\ref{thm:unweighted-exchange}, and the next breakpoint can also be computed in strongly polynomial time. Since the initial partition can be found by a strongly polynomial matroid-partition algorithm, the entire algorithm is strongly polynomial.
\end{proof}

\section{Applications}
\label{sec:applications}

We give two applications of weighted equitability: the first concerns balanced schedules of matroid-constrained jobs on identical machines, while the second gives fair division under a common matroid constraint.

\subsection{Matroid-constrained load balancing}
\label{sec:load}

We consider the following scheduling problem on $m$ identical machines. The jobs form the ground set $J$ of a matroid $M=(J,\cI)$, and job $e\in J$ has processing time $p_e\in[0,1]$. A feasible schedule is a partition $J=I_1\sqcup\dots\sqcup I_m$ such that $I_i\in\cI$ for every $i\in[m]$. The \emph{load} of machine $i$ is $p(I_i)=\sum_{e\in I_i}p_e$, and the \emph{makespan} is the $\max_{i\in[m]}p(I_i)$. Since the machines are identical and there are no release dates or precedence constraints, the order of the jobs assigned to a machine is irrelevant for the makespan. This is a min--max matroid partitioning problem; see~\cite{kawase2021optimal}.

Our main theorem gives a schedule whose makespan is within an additive error $(1-\frac{1}{m})p_{\max}$ of the optimum, where $p_{\max}\coloneqq\max_{e\in J}p_e$.

\begin{theorem}\label{thm:schedule}
Let $M=(J,\cI)$ be a matroid whose ground set can be partitioned into $m$ independent sets, and let $p\in\R_{\ge 0}^J$. There is a polynomial time algorithm for matroid-constrained makespan minimization for identical machines that computes a feasible schedule $J=I_1\sqcup\dots\sqcup I_m$ such that the makespan
\[
\max_{i\in[m]}p(I_i)\leq \opt+(1-\frac{1}{m})p_{\max}.
\]
\end{theorem}

\begin{proof}
    Let $r=r(M)$ and set $d=mr-|J|$. Note that $d\geq 0$ by the assumption that $J$ can be partitioned into $m$ independent sets. Add a set $D$ of $d$ dummy elements, each of processing time zero, and define
    \[
        N\coloneqq (M\oplus U_{d,d})_r.
    \]
    By the definition of truncation,
    \begin{equation}\label{eq:dummy-rank}
    r_N(A)=\min\{r,r_M(A\cap J)+|A\cap D|\}
    \qquad\forall A\subseteq J\cup D.
    \end{equation}
    In particular, $N|_J=M$. Fix any feasible partition $J=I_1\sqcup\cdots\sqcup I_m$. Since each $I_i$ is independent, $|I_i|\le r$, and
    \[
      \sum_{i=1}^m(r-|I_i|)=mr-|J|=d.
    \]
    We may therefore partition $D=D_1\sqcup\cdots\sqcup D_m$ so that $|D_i|=r-|I_i|$ for every $i$. Then $I_i\cup D_i$ is a basis of $N$, and hence $J\cup D$ can be partitioned into $m$ bases of $N$.
    
    Apply Theorem~\ref{thm:weighted} to $N$, extending $p$ by zeros on $D$. We obtain a partition
    \[
      J\cup D=B_1\sqcup\cdots\sqcup B_m
    \]
    into bases of $N$ such that
    \[
    \max_{i\in[m]}p(B_i) \leq \frac{p(J\cup D)}{m}+\left(1-\frac{1}{m}\right)p_{\max}.
    \]
    Let $I_i=B_i\cap J$. 
    Since $N|J=M$, each $I_i$ is independent in $M$, and the $I_i$ form a partition of $J$. The dummy elements have processing time zero, so $p(I_i)=p(B_i)$. Therefore,
    \[
    \max_{i\in[m]}p(I_i) \leq \frac{p(J)}{m}+\left(1-\frac{1}{m}\right)p_{\max} \leq \opt+\left(1-\frac{1}{m}\right)p_{\max},
    \]
    where the second inequality follows from the lower bound $\opt\geq p(J)/m$. The running time then follows: a feasible partition can be found by matroid partition, and independence in $N$ can be tested using the rank of the truncated matroid.
\end{proof}
The following is a direct corollary of Theorem \ref{thm:schedule} and the fact that $p_{\max}\le \opt$.
\begin{corollary}
    There is a $(2-\frac{1}{m})$-approximation algorithm for matroid-constrained makespan minimization for $m$ identical machines.
\end{corollary}

\subsection{Fair division}
\label{sec:fair}

We next interpret the elements as indivisible goods and consider envy-free allocations under a common matroid constraint. Let $M=(E,\cI)$ be a matroid. There are $k$ agents, each with a valuation function $v_i\colon 2^E\to\R_{\geq 0}$. A \emph{feasible allocation} $\mathcal{P}=\{B_i\}_{i=1}^k$ is a partition of $E$ into $k$ bases $E=B_1\sqcup \cdots \sqcup B_k$. The allocation is  \emph{envy-free up to one good (EF1)} if for every $i,j\in [k]$, there exists a good $e\in B_j$ such that $v_i(B_i)\geq v_i(B_j-e)$. 

We consider the special case of identical additive valuations. In this setting, there is a weight vector $w\in \R_{\ge 0}^E$ such that $v_i(X)=\sum_{e\in X}w_e$ for every $X\subseteq E$ and $i\in[k]$. Then, an allocation is EF1 if for every $i,j\in[k]$, there exists a good $e\in B_j$ such that $w(B_i)\geq w(B_j-e)$. 

Matroid-constrained EF1 allocations were first studied by Biswas and Barman~\cite{biswas2018fair}. They gave a polynomial-time algorithm for arbitrary additive valuations under partition matroid constraints and, in the case of identical additive valuations, under laminar matroid constraints. Dror, Feldman, and Segal-Halevi~\cite{dror2023fair} proved the existence of EF1 allocations under identical base-orderable matroid constraints for two agents with additive valuations and for three agents with binary additive valuations. More recently, Akrami, Liu, Raj, and V\'egh~\cite{AkramiLiuRajVegh2026} established the existence of EF1 allocations for arbitrary matroids with identical tri-valued additive valuations. Oki and Schwarcz~\cite{oki2025generalizing} obtained an EF1 allocation for two agents under a matroid representable over a field of characteristic zero, assuming that at least one of the two valuations is additive. 

The following theorem establishes EF1 for an arbitrary number of agents, arbitrary matroids, and identical additive valuations.

\begin{theorem}
    Let $M=(E,\cI)$ be a matroid whose ground set can be partitioned into $k$ bases, and let $w\in\R_{\ge 0}^E$. Then there exists a matroid-constrained EF1 allocation, and such an allocation can be found in strongly polynomial time.
\end{theorem}

\begin{proof}
    Let $E=\bar{B}_1\sqcup \cdots \sqcup \bar{B}_k$ be the final partition returned by the algorithm for $\tau=\max_{e\in E}w_e$. Fix $j\in[k]$, and denote by $W_j\coloneqq \max_{e\in \bar{B}_j} w_e$. 
    
    The first key observation is that $B_j$ is never involved in a repair step for $\tau\ge W_j$. Suppose not, and consider the last such repair step. Let  $E=B_1\sqcup \cdots \sqcup B_k$ be the partition before this repair step. Let $S$ be the active set, and let $b_i=|B_i\cap S|\ \forall i\in [k]$. Suppose we perform a repair step on $B_{j_1}$ and $B_{j_2}$ where $j\in \{j_1,j_2\}$ and $b_{j_1}-b_{j_2}>1$. Then, after repairing, $b'_{j_1}=b_{j_1}-1>b_{j_2}\ge 0$, and $b'_{j_2}=b_{j_2}+1>0$. Therefore, both $B_{j_1}'$ and $B_{j_2}'$ have nonempty intersection with $S$. Since this is the last repair step involving $B_j$, this implies that $\bar{B}_j\cap S\neq \emptyset$, which is a contradiction to the fact that $\max_{e\in \bar{B}_j} w_e=W_j\le \tau$.

    Another key observation is that $\min_{i\in [k]} w^\tau(B_i)$ is non-decreasing throughout the algorithm. Indeed, when $\tau$ increases, $w^\tau$ is non-decreasing and thus $\min_{i\in [k]} w^\tau(B_i)$ is non-decreasing. When we perform a repair step on $B_i,B_j$ where $b_i-b_j>1$, by \eqref{eq:diff-before} and \eqref{eq:diff-after},
    \[|w^\tau(B'_i)-w^\tau(B'_j)|\le |w^\tau(B_i)-w^\tau(B_j)|.\]
    It follows from the fact that $w^\tau(B'_i)+w^\tau(B'_j)=w^\tau(B_i)+w^\tau(B_j)$ that 
    \[\min\{w^\tau(B'_i), w^\tau(B'_j)\}\ge \min\{w^\tau(B_i), w^\tau(B_j)\}.\] 
    Since the weights of all other bases remain unchanged, $\min_{i\in [k]} w^\tau(B_i)$ is non-decreasing.

    Finally, combine the two observations and \eqref{eq:equitable_leve_tau}. Let $E=B_1\sqcup \cdots \sqcup B_k$ be the partition when $\tau=W_j$. We have
    \[
    \min_{i\in [k]} w(\bar{B}_i)\ge \min_{i\in [k]} w^{W_j} (B_i)\ge w^{W_j}(B_j)-W_j=w(\bar{B}_j)-\max_{e\in \bar{B}_j}w_e,
    \]
    where the first inequality follows from the monotonicity of $\min_{i\in [k]} w^\tau(B_i)$; the second inequality follows from \eqref{eq:equitable_leve_tau}; the equality follows from the first observation and the definition of $W_j$. Let $e\in\bar{B}_j$ be such that $w_e=W_j$. Then $w(\bar{B}_i)\geq w(\bar{B}_j-e)$ for every $i\in[k]$. Since $j$ was arbitrary, the allocation is EF1.
\end{proof}

\section{Weighted degree-bounded bases}
\label{sec:degree_bounded}

The weighted degree-bounded basis theorem proved in this section may be viewed as a matroidal analogue of the Beck--Fiala theorem for discrepancy minimization \cite{beck1981integer}. The original theorem is stated for $x=\frac{1}{2}\cdot \ind$ and $W\in \{0,1\}^{m\times n}$, with the discrepancy bound $\Delta-\frac{1}{2}$. Their proof can be easily generalized to prove the following theorem.

\begin{theorem}[Beck-Fiala]\label{thm:Beck--Fiala}
    Let $W\in \R_{\ge 0}^{m\times n}$ be a matrix such that there are at most $\Delta$ nonzeros per column. Then, for every $x\in [0,1]^n$, there exists $y\in \{0,1\}^n$ such that
    \[
    \inftynorm{Wy-Wx}\le \Delta\cdot\max_{i,j}W_{i,j}.
    \]
\end{theorem}

We consider the matroid-constrained discrepancy minimization. Given $x\in P(M)$ where $P(M)$ is the base polytope of matroid $M$, we show that one can find a basis $B$ of $M$ whose indicator vector $y$ satisfies the guarantee in Theorem~\ref{thm:Beck--Fiala} within a factor of $2$. In fact, we obtain the min-cost version of the same guarantee. 

This can also be viewed as a weighted analogue of a theorem of Kir\'aly, Lau, and Singh \cite{kiraly2008degree}, which proves a slightly stronger bound $(2\Delta-1)$ for the special case where $W\in \{0,1\}^{m\times n}$. Their result is stated for degree-bounded matroid bases in hypergraphs. In that case, $W$ is the incidence matrix of a hypergraph, where every row corresponds to a hyperedge. The column sparsity of $W$ implies the degree of the hypergraph is at most $\Delta$. They prove that given $x\in P(M)$ and costs $c\in \R_{\ge 0}^E$, one can find a basis $B$ of $M$ such that $c(B)\le c^\top x$, and for every hyperedge $\mathscr{E}$, $\left||B\cap \mathscr{E}|-x(\mathscr{E})\right|\le 2\Delta-1$. We prove the following weighted version of degree-bounded bases:



\begin{theorem}\label{thm:weighted_degree_bounded_basis}
    Let $M=(E,\mathcal{I})$ be a matroid, let $W\in \R_{\ge 0}^{m\times E}$ be a matrix such that there are at most $\Delta$ nonzeros per column, and let $c\in\R^E$ be the costs. For every $x\in P(M)$, there is a polynomial time algorithm to compute a basis $B$ of $M$ such that $c(B)\le c^\top x$ and
\[\inftynorm{W\ind_B-Wx}\leq 2\Delta\cdot \max_{i,e} W_{i,e}.\]
\end{theorem}
The proof generalizes the iterative rounding approach in \cite{kiraly2008degree}, which repeatedly fixes integral variables in an extreme-point solution and drops any row containing at most $2\Delta$ remaining variables. A token-counting argument shows that one of these operations is always possible. We delay it to Appendix \ref{sec:proof-degree-bases}.

Theorem \ref{thm:weighted_degree_bounded_basis} generalizes discrepancy rounding results for the generalized assignment problem. Let $I$ be a set of $m$ machines and $J$ be a set of $n$ jobs. Job $j$ has a processing time $d_{ij}$ on machine $i$. Let $G=(I\cup J,E)$ be a bipartite graph. For convenience, we define $x_{ij}=y_{ij}=d_{ij}=0$ for every $(i,j)\notin E$. We say $x\in [0,1]^{E}$ is a \emph{fractional assignment} if $\sum_{i=1}^m x_{ij}=1$ for every $j\in[n]$. We say $y\in\{0,1\}^{E}$ is an \emph{(integral) assignment} if  $\sum_{i=1}^m y_{ij}=1$. The goal is to find an assignment such that the load on each machine is close to the load of the fractional assignment. 
\begin{corollary}\label{cor:LST-2}
    Let $G=(I\cup J,E)$ be a bipartite graph. For every fractional assignment $x\in [0,1]^{E}$, $d\in \R_{\ge 0}^{E}$, and costs $c\in\R^{E}$, there exists an assignment $y\in\{0,1\}^{E}$ such that $c^\top y\leq c^\top x$ and
    \begin{equation}\label{eq:LST-2}
        \left|\sum_{j=1}^nd_{ij}y_{ij}-\sum_{j=1}^nd_{ij}x_{ij}\right|\leq 2\cdot \max_{e\in E} d_e\qquad \forall i\in [m].
    \end{equation}
\end{corollary}
\begin{proof}
    Consider the partition matroid $M=(E,\mathcal{I})$ with the ground set $E$, and partition classes $P_1,\dots,P_n$ where $P_j=\{(i,j)\in E\colon i\in [m]\}$. Let $W\in \R_{\ge 0}^{I\times E}$ be the constraint matrix encoding the processing times. Namely, the $i$-th row and $e$-th column is $d_e$ if edge $e$ is incident to machine $i$, and $0$ otherwise. Note that every edge is only incident to one machine, and thus $W$ has at most $\Delta=1$ nonzeros per column. The fractional assignment $x$ belongs to $P(M)$. Therefore, Theorem \ref{thm:weighted_degree_bounded_basis} gives a basis $B$ of $M$ with incidence vector $y$ such that $c^\top y \leq c^\top x$ and $|\sum_{j=1}^n d_{ij}y_{ij}-\sum_{j=1}^n d_{ij}x_{ij}|\leq 2\cdot\max_{e\in E} d_e$ for every $i\in [m]$.
\end{proof}
A seminal result by Lenstra, Shmoys, and Tardos \cite{lenstra1990approximation} proves a one-sided bound. Namely, for every fractional assignment $x$ and $d$, there exists an integral assignment $y$ such that
\[
    \sum_{j=1}^n d_{ij}y_{ij}-\sum_{j=1}^n d_{ij}x_{ij}\leq \max_{e\in E} d_e\qquad \forall i\in[m].
\]
Their result is stated for an extreme point $x$ of the LP relaxation, but this is without loss of generality. It implies a $2$-approximation algorithm for \emph{makespan minimization}, where the goal is to find an assignment that minimizes the max load of the $m$ machines. Later on, Shmoys and Tardos \cite{shmoys1993approximation}  generalized the above result to the min-cost version. They use an ordered-bucketing approach to reduce to a min-cost perfect matching problem in an auxiliary graph. In fact, it is not hard to generalize their proof to prove a two-sided bound
\begin{equation}\label{eq:LST}
    \left|\sum_{j=1}^n d_{ij}y_{ij}-\sum_{j=1}^n d_{ij}x_{ij}\right|\le \max_{e\in E} d_e\quad\forall i\in [m],
\end{equation}
which is slightly stronger than \eqref{eq:LST-2}. Interestingly, this also follows as a consequence of Budish, Che, Kojima, and Milgrom \cite{budish2013designing} (Theorem 9) in the context of designing allocation mechanisms. They prove that every fractional $x$ can be written as a convex combination of integral assignments $y$ satisfying \eqref{eq:LST}, which implies one such $y$ has cost $c^\top y\le c^\top x$.

\section{Prefix-constrained bases}
\label{sec:prefix}
We propose the matroidal analogue of the prefix GAP Conjecture \ref{conj:prefix_ssuf} on prefix-constrained bases.

\begin{conjecture}[Prefix-constrained bases]\label{conj:matroid_chain}
    Let $M=([n],\cI)$ be a matroid, and let $w\in \R_{\ge 0}^n$ be weights. For every $x\in P(M)$, there exists a basis $B$ of $M$ such that for some universal constant $C>0$,
\begin{equation}\label{eq:prefix-bases}
\left|\sum_{i=1}^tw_i\mathbf{1}\{i\in B\}-\sum_{i=1}^tw_ix_i\right|\leq C\cdot \max_{i\in [n]} w_i\quad \forall t\in [n].
\end{equation}
\end{conjecture}
\begin{lemma}
    Conjecture \ref{conj:matroid_chain} implies Conjecture \ref{conj:prefix_ssuf} up to a factor of $2$.
\end{lemma}
\begin{proof}
    Let $G=(I\cup J,E)$ be a bipartite graph with $|I|=m$ and $|J|=n$. We take the partition matroid $M$ on the ground set $E$ with partition classes $P_1,\dots,P_n$ where $P_j=\{(i,j)\in E\colon i\in [m]\}$. The elements in the ground set $E$ are ordered lexicographically by $\{(i,j)\in E:i\in [m], j\in [n]\}$. In other words, for every machine $i$, the edges $E_i\coloneqq \{(i,j)\in E:j\in [n]\}$ incident to it appear consecutively in the ordering, which are then ordered locally by $j$. Therefore, for every $t\in [n]$, the prefix of jobs $E_i^t\coloneqq \{(i,j)\in E:j\in [t]\}$ incident to machine $i$ appears as an interval in the ordering. Note that $x\in P(M)$. Conjecture \ref{conj:matroid_chain} implies for every $1\le a\le b\le |E|$,
$$\left|\sum_{i=a}^bw_i\Big(\mathbf{1}\{i\in B\}-x_i\Big)\right|\le \left|\sum_{i=1}^{a-1}w_i\Big(\mathbf{1}\{i\in B\}-x_i\Big)\right|+\left|\sum_{i=1}^bw_i\Big(\mathbf{1}\{i\in B\}-x_i\Big)\right| \leq 2C\cdot \max_{e\in E} w_e.$$
    Applying Conjecture \ref{conj:matroid_chain} to matroid $M$ and weights $w\coloneqq d\in \R_{\ge 0}^E$, we obtain a basis $B$ whose indicator vector $y$ defines an assignment. Choosing the interval $[a,b]=E_i^t=\{(i,j)\in E:j\in [t]\}$, we obtain
    \[
    \left|\sum_{j=1}^t d_{ij}y_{ij}-\sum_{j=1}^t d_{ij}x_{ij}\right| \leq 2C\cdot \max_{e\in E} d_e. \qedhere
    \]
\end{proof}

We also consider the following conjecture on cost-augmented prefix-constrained bases.

\begin{conjecture}\label{conj:matroid_chain-cost}
    Let $M=([n],\cI)$ be a matroid, let $w\in \R_{\ge 0}^n$ be weights, and let $c\in\R_{\ge 0}^n$ be costs. For every $x\in P(M)$, there exists a basis $B$ of $M$ such that $c(B)\leq c^\top x$, and for some universal constant $C>0$, $$\left|\sum_{i=1}^tw_i\mathbf{1}\{i\in B\}-\sum_{i=1}^tw_ix_i\right|\leq C\cdot \max_{e\in[n]} w_e\quad \forall t\in [n].$$
\end{conjecture}

Swamy, Traub, Koch, and Zenklusen \cite{swamy2026unsplittable} developed a general framework to translate discrepancy rounding results to their cost-augmented version. Using their framework, we obtain the following lemma; its proof is in Appendix \ref{sec:prefix-cost-proof}.
\begin{lemma}\label{lem:prefix-cost}
    Conjecture \ref{conj:matroid_chain} implies Conjecture \ref{conj:matroid_chain-cost} up to a factor of $2$.
\end{lemma}

We give an algorithm that proves a bound of $C=O(\log n)$ for Conjecture \ref{conj:matroid_chain-cost}. The algorithm iteratively solves an LP relaxation that drops a constant fraction of the prefix constraints. We show that every iteration finds a partial coloring where a constant fraction of the variables are fixed. After $O(\log n)$ iterations, the algorithm terminates, and the discrepancy bound follows.
\begin{theorem}
    Let $M=([n],\cI)$ be a matroid, let weights $w\in \R_{\ge 0}^n$, and let costs $c\in \R_{\ge 0}^n$. For every $x\in P(M)$, there exists a basis $B$ of $M$ such that $c(B)\le c^\top x$ and
\[
\left|\sum_{i=1}^tw_i\mathbf{1}\{i\in B\}-\sum_{i=1}^tw_ix_i\right|\leq O(\log n)\cdot \max_{e\in [n]} w_e\quad \forall t\in [n].
\]
\end{theorem}
\begin{proof}

   Let $b_t\coloneqq \sum_{i=1}^t w_ix_i$. Consider the natural LP relaxation for min-cost prefix-constrained bases, on top of which we further relax by keeping every fourth prefix constraint:
   \begin{equation}\label{eq:prefix_matroid-log}
        \begin{aligned}
            \min&\ \sum_{e\in [n]} c_ey_e\\
            s.t.\ 
            &y(S)\leq r(S)\quad &\forall S\subseteq [n]\\
            &y([n])=r([n])\\
            &\sum_{i=1}^tw_iy_i=b_t\quad&\forall t\in [n],\ t=4k\text{ for some }k\in \Z\\
            &0\leq y_e\leq 1\quad &\forall e\in [n].
        \end{aligned}
    \end{equation}
    We claim that the optimal extreme point $y^*$ of this relaxation has at least $\frac{n}{6}$ integral coordinates. To see this, 
    let $I\subseteq [n]$ be the integral coordinates and $F=[n]\setminus I$ be the fractional coordinates of $y^*$. Suppose $|I|=m$. Let $\mathcal{T}=\{S\subseteq [n]\colon y^*(S)=r(S)\}$ be the collection of all tight sets from the matroid constraints. By a standard uncrossing argument, there is a chain $\mathcal{L}=\{S_1,\dots,S_\ell\}$ where $S_1\subsetneq\cdots\subsetneq S_\ell$ such that $\spa(\{\chi_S\colon S\in\mathcal{L}\})=\spa(\{\chi_S\colon S\in\mathcal{T}\})$. Set $S_0\coloneqq\emptyset$. For every $h\in[\ell]$, one has $y^*(S_h\setminus S_{h-1})=y^*(S_h)-y^*(S_{h-1})=r(S_h)-r(S_{h-1})\in\Z$. If $(S_h\setminus S_{h-1})\cap F\neq\emptyset$, there is at least one $e\in S_h\setminus S_{h-1}$ with $0<y^*_e<1$, and then $y^*(S_h\setminus S_{h-1})>0$; hence, $y^*(S_h\setminus S_{h-1})\geq 1$. The fact that $y_e^*<1$ further implies $|S_h\setminus S_{h-1}|\geq 2$. Otherwise, $|S_h\setminus S_{h-1}|\geq 1$. Therefore, 
    \[
    n=\sum_{h=1}^\ell |S_h\setminus S_{h-1}|\ge 2\ell-\sum_{h=1}^\ell \ind\{(S_h\setminus S_{h-1})\cap F=\emptyset\}\ge 2\ell-\sum_{h=1}^\ell \ind\{(S_h\setminus S_{h-1})\cap I\neq\emptyset\}\ge 2\ell-m.
    \]  
    Therefore, the number of independent tight constraints is at most 
    \[
    \ell+\frac{n}{4}+m\le \frac{m+n}{2}+\frac{n}{4}+m=\frac{3m}{2}+\frac{3n}{4},
    \] 
    which consists of $\ell$ tight matroid constraints from $\cal{L}$, $\frac{n}{4}$ prefix constraints, and $m$ tight box constraints $0\le y_e\le 1$ for $e\in I$. Since $y^*$ is an extreme point, the number of independent tight constraints is at least $n$. Therefore,
    \[
    n\le \frac{3m}{2}+\frac{3n}{4},
    \]
    which implies $m\ge \frac{n}{6}$.

    The algorithm iteratively solves \eqref{eq:prefix_matroid-log} to obtain an extreme point $y^*$ and fix all variables with $y^*_e\in\{0,1\}$. In the matroid, it is equivalent to contracting elements with $y^*_e=1$ and deleting elements with $y^*_e=0$. We update the ground set by renumbering the unfixed elements in the same order and update $b_t\gets \sum_{i=1}^t w_iy^*_i$; solve the updated LP relaxation. Since every iteration fixes at least a $\frac{1}{6}$ fraction of the variables, the number of variables reduces by a constant factor. The algorithm terminates in $O(\log n)$ iterations. In the end we obtain an integral solution corresponding to a basis $B$ of $M$. The cost guarantee follows because we only relax the LP. 
    
    To prove the discrepancy bound, fix an arbitrary prefix $t$ and an iteration $j$. Denote by $y^{(j)}$ the optimal solution in iteration $j$. We claim
    \begin{equation}\label{eq:log-two-iteration}
    \left|\sum_{i=1}^t w_iy^{(j)}_i-\sum_{i=1}^t w_iy^{(j-1)}_i\right|\le 6\cdot \max_{e\in [n]} w_e.
    \end{equation}
    To see this, let $b^{(j-1)}_t\coloneqq \sum_{i=1}^t w_iy^{(j-1)}_i$. Then the $j$-th iteration solves the LP \eqref{eq:prefix_matroid-log} with the right-hand side $b^{(j-1)}_t$. Since we keep every one out of four consecutive prefix constraints in the LP, there is a $t'$ whose constraint is kept and $|t'-t|\le 3$. Therefore,

    \[
        \left|\sum_{i=1}^t w_iy^{(j)}_i-\sum_{i=1}^t w_iy^{(j-1)}_i\right|=\left|\Big(\sum_{i=1}^t w_iy^{(j)}_i-\sum_{i=1}^{t'} w_iy^{(j)}_i\Big)+\Big(\sum_{i=1}^t w_iy^{(j-1)}_i-\sum_{i=1}^{t'} w_iy^{(j-1)}_i\Big)\right|\le 6\cdot \max_{e\in [n]} w_e,
    \]
    where the first inequality follows from the fact that $\sum_{i=1}^{t'} w_iy^{(j)}_i=\sum_{i=1}^{t'} w_iy^{(j-1)}_i$; the second inequality follows from the fact that $|t-t'|\le 3$ and $0\le y_i\le 1$. This proves the claimed bound \eqref{eq:log-two-iteration} between two consecutive iterations. Summing over $O(\log n)$ iterations gives 
\[
\left|\sum_{i=1}^tw_i\mathbf{1}\{i\in B\}-\sum_{i=1}^tw_ix_i\right|\leq O(\log n)\cdot \max_{e\in [n]} w_e,
\]
as claimed.
\end{proof}

Below we show that Conjecture \ref{conj:matroid_chain-cost} is true with $C=1$ for two special cases: for binary weights $w\in \{0,1\}^n$ or for uniform matroids. Their proofs can be found in Appendices \ref{sec:proof-binary} and \ref{sec:proof-uniform}, respectively.
\begin{prop}\label{prop:binary-weights}
    Conjecture \ref{conj:matroid_chain-cost} is true with $C=1$ for $w\in \{0,1\}^n$.
\end{prop}
\begin{theorem}\label{thm:uniform}
    Let $M=([n],\cI)$ be a uniform matroid and let weights $w\in \R_{\ge 0}^n$. For every $x\in P(M)$, there exists a basis $B$ of $M$ such that $$\left|\sum_{i=1}^tw_i\mathbf{1}\{i\in B\}-\sum_{i=1}^tw_ix_i\right|\leq \max_{e\in [n]} w_e\quad \forall t\in [n].$$
\end{theorem}


\section{Single-source unsplittable flows and counterexamples}
\label{sec:unsplittable}

We disprove Conjecture~\ref{conj:weighted-carpool}, and hence also Conjecture~\ref{conj:ssuf}. We begin with a simpler counterexample directly to Conjecture~\ref{conj:ssuf}.

\begin{theorem}\label{thm:counterexample-MS}
For each $k \ge 4$, the construction in Figure~\ref{fig:unsplittable-flow-counterexample} is a counterexample to Conjecture~\ref{conj:ssuf} with $\max_{j\in [n]}d_j=8k$, and for every  unsplittable flow $y$,
\[
 \max_a |x_a-y_a|\ge 9k-3
 =\left(\frac{9}{8} -\frac{3}{8k}\right)\cdot \max_{j\in [n]}d_j.
\]
\end{theorem}
\begin{figure}[t]
\centering
\resizebox{\linewidth}{!}{%
\begin{tikzpicture}[
  scale=0.80,
  transform shape,
  flow/.style={
    draw=black!25,
    line width=0.55pt,
    ->,
    >=latex
  },
  selected flow/.style={
    flow,
    draw=black,
    line width=1.35pt
  },
  crossing/.style={
    flow,
    preaction={draw=white,line width=2.4pt,-}
  },
  selected crossing/.style={
    selected flow,
    preaction={draw=white,line width=3.4pt,-}
  },
  flow label/.style={
    fill=white,
    fill opacity=0.97,
    text opacity=1,
    inner xsep=1.2pt,
    inner ysep=0.7pt,
    text=black!75,
    font=\small
  },
  junction/.style={
    circle,
    fill=black,
    draw=black,
    inner sep=0pt,
    minimum size=1.7mm
  },
  terminal/.style={
    rectangle,
    draw=red!70!black,
    fill=red!8,
    text=red!70!black,
    line width=0.9pt,
    minimum width=8mm,
    minimum height=6mm,
    inner sep=0pt,
    font=\large
  },
  source/.style={
    circle,
    draw=green!50!black,
    fill=green!25,
    text=green!40!black,
    line width=1pt,
    minimum size=6.2mm,
    inner sep=0pt,
    font=\small
  }
]

\def\flowgadget#1#2#3#4{%
  \begin{scope}[shift={(#2,0)}]

    \node[junction] (u#1) at (-3, 3.10) {};
    \node[junction] (v#1) at ( 3, 3.10) {};
    \node[junction] (eu#1) at (-3,-1.15) {};
    \node[junction] (ev#1) at ( 3,-1.15) {};

    \node[terminal] (ta#1) at (-2,1.68) {$8k$};
    \node[terminal] (tc#1) at ( 0,1.68) {$\ 6k-2\ $};
    \node[terminal] (tb#1) at ( 2,1.68) {$8k$};

    \node[junction] (p#1) at (-1,0.05) {};
    \node[junction] (q#1) at ( 1,0.05) {};

    \draw[#3] (s) -- node[flow label,pos=0.61] {$15k+3$} (u#1);
    \draw[#3] (s) -- node[flow label,pos=0.61] {$15k+3$} (v#1);

    \draw[#3] (u#1) -- node[flow label,pos=0.72] {$13k+1$} (eu#1);
    \draw[#3] (v#1) -- node[flow label,pos=0.72] {$13k+1$} (ev#1);

    \draw[flow] (u#1) -- node[flow label,pos=0.43] {$2k+2$} (ta#1);
    \draw[#3]   (v#1) -- node[flow label,pos=0.43] {$2k+2$} (tb#1);

    \draw[#3]   (p#1) -- node[flow label,pos=0.65] {$6k-2$} (ta#1);
    \draw[flow] (q#1) -- node[flow label,pos=0.65] {$6k-2$} (tb#1);

    \draw[flow] (p#1) -- node[flow label,pos=0.4] {$3k-1$} (tc#1);
    \draw[#3]   (q#1) -- node[flow label,pos=0.4] {$3k-1$} (tc#1);

    \draw[#3] (eu#1) -- node[flow label,pos=0.25] {$9k-3$} (q#1);
    \draw[#4] (ev#1) -- node[flow label,pos=0.25] {$9k-3$} (p#1);

    \draw[#3] (eu#1) -- node[flow label,pos=0.43] {$4k+4$} (z);
    \draw[#3] (ev#1) -- node[flow label,pos=0.43] {$4k+4$} (z);
  \end{scope}%
}

\node[source] (s) at (0,5.35) {$s$};
\node[terminal] (z) at (0,-3.10) {$8k$};

\flowgadget{L}{-7.20}{selected flow}{selected crossing}
\flowgadget{R}{ 7.20}{flow}{crossing}

\node[font=\Large] at (0,1.15) {$\cdots$};

\node[terminal] (rL) at (-2,-3.10) {$8k$};
\node[terminal] (rR) at ( 2,-3.10) {$8k$};
\node[font=\Large] at ( 1,-3.1) {$\cdots$};
\node[font=\Large] at (-1,-3.1) {$\cdots$};
\node[font=\Large] at (0,-4)
  {$\underbrace{\qquad\qquad\qquad\qquad}_{k+1}$};

\end{tikzpicture}%
}
\caption{The fractional flow used in the counterexample to the Morell--Skutella conjecture. Arc labels indicate fractional flow values, and terminal labels indicate demands. Bold arcs in the first gadget show the paths forced to be chosen in the proof of Theorem \ref{thm:counterexample-MS}.}
\label{fig:unsplittable-flow-counterexample}
\end{figure}
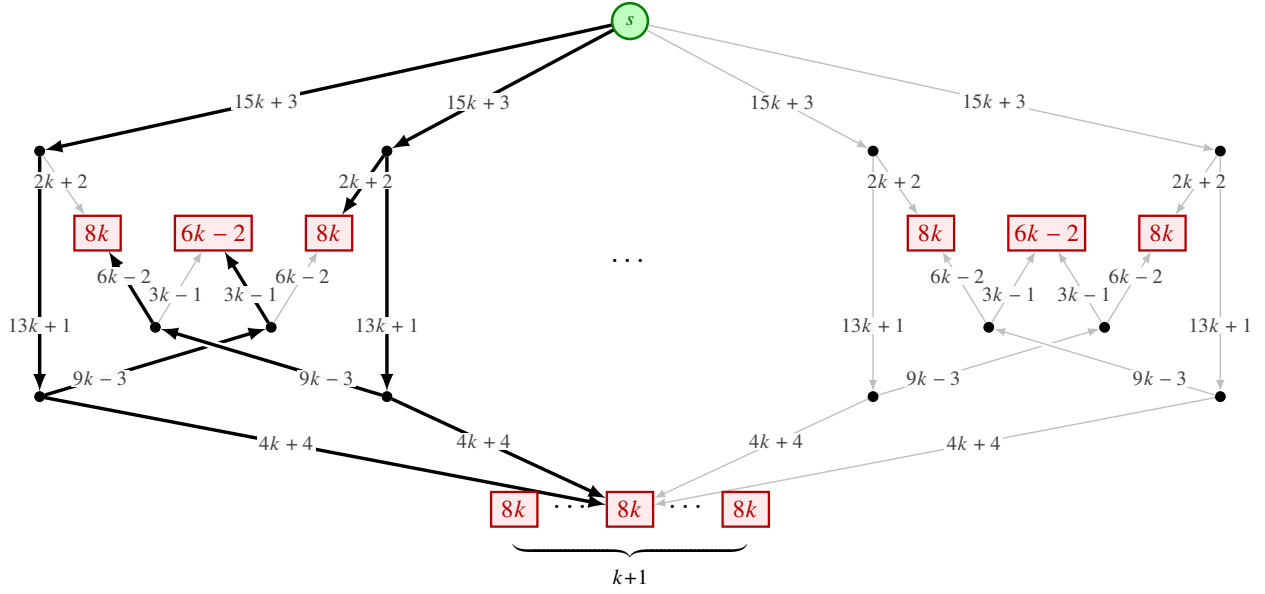

\begin{proof}
In the construction presented in Figure~\ref{fig:unsplittable-flow-counterexample}, we have $k$ identical gadgets with three terminals each, and additionally $k+1$ terminals of demand $8k$ each, which we call \emph{central terminals}. Suppose for contradiction that there exists an unsplittale flow $y$ such that $|x_a-y_a|\le 9k-4$ for all $a\in A$. In an unsplittable flow, every terminal chooses one path from the source $s$, which we call an $s$-path. First note that each $s$-path to central terminals can only be used at most once, since otherwise an arc $a$ of fractional flow $x_a=4k+4$ would have $y_a-x_a\ge 16k - (4k+4) = 12k-4 > 9k-4$, a contradiction. 

Next, we claim that at most one central terminal can go through each gadget. This would lead to a contradiction as we have $k+1$ central terminals but only $k$ gadgets. Suppose for contradiction that on some gadget, both left and right $s$-paths to central terminals are used. Each of the three terminals inside the gadget has two $s$-paths to choose from; we say it is the left (right) $s$-path if it overlaps with the left (right) $s$-path to the central terminals. By symmetry, we may assume that the terminal of demand $6k-2$ inside the gadget is routed through the left $s$-path. If the terminal of demand $8k$ to its right is also routed through the left $s$-path, the left arc of fractional flow $x_a=13k+1$ would have $y_a=8k+8k+(6k-2) = 22k-2$. Then, $y_a-x_a=(22k-2)-(13k+1) = 9k-3$, contradicting our assumption. So the terminal of demand $8k$ to the right must use the right $s$-path. Further, the right arc of fractional flow $x_a=9k-3$ must be used at least once, as otherwise $x_a-y_a=x_a=9k-3$, a contradiction. It follows that the terminal of demand $8k$ on the left must use the right $s$-path. But then the right arc of fractional flow $x_a=15k+3$ has $y_a=8k+8k+8k=24k$, which makes $y_a-x_a=24k-(15k+3)=9k-3$, a contradiction.
\end{proof}

Now we describe a counterexample to Conjecture~\ref{conj:weighted-carpool}. Recall that a fractional assignment $x\in [0,1]^E$ is defined on the edges of a bipartite graph $G=(I\cup J, E)$ where $|I|=m$ and $|J|=n$. For convenience, in the following, a fractional assignment is given as a matrix $x\in [0,1]^{m\times n}$. The bipartite graph is the support of $x$, i.e., $(i,j)\in E$ if and only if $x_{ij}\neq 0$. Thus we require any (integral) assignment $y$ satisfies $y_{ij}=0$ if $x_{ij}=0$.

\begin{theorem}\label{thm:counterexample-carpool}
For any positive integer $k \ge 4$ there exists a fractional assignment $x \in [0,1]^{(3k+1) \times 6k}$ and weights $d \in \R^{6k}_{\ge 0}$, so that $\max_{j}d_j=10k+3$, and for every assignment $y \in \{0,1\}^{(3k+1) \times 6k}$,   
\[
    \max_{i,t}\left|\sum_{j=1}^t d_{j}y_{ij}-\sum_{j=1}^t d_{j}x_{ij}\right|\ge 11k=\frac{11k}{10k+3} \cdot \max_{j}d_j.
\]
\end{theorem}

\begin{proof}
The fractional assignment $x$ is given as a $3k+1$ by $6k$ matrix such that every column sum is $1$. An integral assignment needs to assign each column $c$ to a row $r$, where $x_{rc}\neq 0$. We label the $3k+1$ rows by $\star$ and $(i,j)$, $i \in [k], j \in [3]$, and the $6k$ columns by $(i,j)$, $i \in [k], j \in [6]$. The first index of a row or column is its block. Each column of $x$ has exactly two nonzero entries which add up to $1$, occurring either in row $\star$ or in the same block. Let $a = 4k-1$ and $b = 3k+2$. Below is the matrix with entries $d_{(i,j)} x_{(i, j), (i',j')}$. Thus, $d_{(i,j)}$ is the sum of the entries on column $(i,j)$, which is
\[
\begin{aligned}
    d_{(i,j)}=\begin{cases}
        a+b\quad &j=2\\
        a+2b\quad &o/w
    \end{cases}.
\end{aligned}
\]
In particular, $\max_{c} d_c=a+2b=10k+3$.
\[
\resizebox{\textwidth}{!}{$
\begin{bNiceArray}{cccccc|c|cccccc}[first-row,first-col]
 & (1,1)&(1,2)&(1,3)&(1,4)&(1,5)&(1,6)
 & \cdots
 & (k,1)&(k,2)&(k,3)&(k,4)&(k,5)&(k,6) \\

\star
 &0&0&3b-a&0&3b-a&0
 &\cdots
 &0&0&3b-a&0&3b-a&0 \\
\hline
(1,1)
 &2b-a&a&2a-b&a+b&0&0
 &\cdots
 &0&0&0&0&0&0 \\
(1,2)
 &2a&0&0&0&0&b
 &\cdots
 &0&0&0&0&0&0 \\
(1,3)
 &0&b&0&b&2a-b&a+b
 &\cdots
 &0&0&0&0&0&0 \\
\hline
\vdots \ \ \ \ \
 &\vdots&\vdots&\vdots&\vdots&\vdots&\vdots
 &\ddots
 &\vdots&\vdots&\vdots&\vdots&\vdots&\vdots \\
\hline
(k,1)
 &0&0&0&0&0&0
 &\cdots
 &2b-a&a&2a-b&a+b&0&0 \\
(k,2)
 &0&0&0&0&0&0
 &\cdots
 &2a&0&0&0&0&b \\
(k,3)
 &0&0&0&0&0&0
 &\cdots
 &0&b&0&b&2a-b&a+b
\end{bNiceArray}
$}
\]
Suppose by contradiction that there exists an assignment $y$ such that
\[
\left|\sum_{c=1}^t d_{c}y_{rc}-\sum_{c=1}^t d_{c}x_{rc}\right|< 11k\qquad \forall r, t.
\]
For row $r$ and column $t$, denote by $D(r,t)\coloneqq\sum_{c=1}^t d_{c}y_{rc}-\sum_{c=1}^t d_{c}x_{rc}$. We claim that, in each block $i \in [k]$, $y_{\star, (i,3)} + y_{\star, (i,5)} \le 1$. First let us show that this leads to a contradiction. Under this assumption, in each block $i \in [k]$, \[ \sum_{j=1}^6 d_{(i,j)} y_{\star, (i,j)} - \sum_{j=1}^6 d_{(i,j)} x_{\star, (i,j)} \le (a+2b) - 2(3b-a) = -4b+3a = -11.\]
Thus, after $k$ blocks, the discrepancy of row $\star$ has $D(\star, (k,6))\le -11k$, a contradiction.

It remains to argue that $y_{\star, (i,3)} + y_{\star, (i,5)} \le 1$. Suppose otherwise, i.e., $y_{\star, (i,3)} = y_{\star, (i,5)} = 1$. If $y_{(i,1), (i,1)} = y_{(i,1),(i,2)} = 1$, then row $(i,1)$ has 
$$D((i,1),(i,2))=(a+2b+a+b)-(2b-a + a) = 2a+b = 11k,$$
a contradiction. If $y_{(i,1), (i,1)} = y_{(i,1),(i,2)} = 0$, then row $(i,1)$ has 
$$D((i,1),(i,3))=0-(2b-a + a + 2a-b) = -11k,$$ which is again not possible. Thus $y_{(i,1), (i,1)} + y_{(i,1),(i,2)} = 1$. 

We now claim $y_{(i,1), (i,4)} = 1$. Otherwise, $y_{(i,3),(i,4)} = 1$. We discuss the following two cases. If $y_{(i,1),(i,1)} = 0$ and $y_{(i,1),(i,2)} =1$, then
$$D((i,1),(i,4))=(a+b) - (2b-a+a+2a-b+a+b)=-(2a+b) = -11k.$$
If $y_{(i,1),(i,1)} = 1$ and $y_{(i,1),(i,2)} =0$ then
$$D((i,3),(i,4))=(a+b+a+2b) - (b +b) = 2a+b = 11k.$$ Neither is possible, so indeed we must have $y_{(i,1), (i,4)} = 1$. 

Next, we claim $y_{(i,1),(i,1)} = 1$ and $y_{(i,1),(i,2)} = 0$, which in turn implies $y_{(i,3),(i,2)} = 1$. Suppose not, which means $y_{(i,1), (i,1)} = 0$, $y_{(i,1),(i,2)} = 1$, and $y_{(i,3),(i,2)} = 0$. Then,
$$D((i,3),(i,5))=0 - (b+b+2a-b) = -(2a+b)=-11k,$$ a contradiction. Thus indeed we must have $y_{(i,1),(i,1)} = 1$ and $y_{(i,3),(i,2)} = 1$. 

Finally, we look at the last column $(i,6)$ of block $i$. If $y_{(i,2),(i,6)} = 0$, then row $(i,2)$ is never selected and has $$D((i,2),(i,6))=0-(2a+b)=-11k.$$ Otherwise, $y_{(i,3),(i,6)} = 0$ and $$D((i,3),(i,6))=(a+b) - (b+b+2a-b+a+b) = -(2a+b)=-11k.$$ In either case, we reach a contradiction, which implies $y_{\star, (i,3)} + y_{\star, (i,5)} \le 1$, as needed. \end{proof}


\section{Conclusion}
We close with some open questions. \begin{itemize}
    \item Can we prove a $C=O(\sqrt{\log n})$ bound for prefix-constrained bases, matching the best known bound for prefix Beck--Fiala? 
    \item Is there a PTAS for matroid-constrained makespan minimization with identical machines, matching the PTAS for the unconstrained version \cite{HochbaumShmoys1987DualApproximation}?
    \item Can we prove the weighted equitability theorem for multiple constraints with column sparsity? Namely, let $M=([n],\mathcal{I})$ be a matroid whose ground set can be partitioned into $k$ bases, and let $W\in \R_{\ge 0}^{m\times n}$ be a matrix which contains at most $\Delta$ nonzeros per column. Is there a partition of $[n]$ into $k$ bases $[n]=B_1\sqcup\cdots\sqcup B_k$ such that for every row $w_i^\top$ of $W$ and $j\in [k]$,
    \[
    \left|w_i(B_j)-\frac{w_i([n])}{k}\right|\le f(\Delta)\cdot \max_{u,v} W_{u,v}
    \]
    for some function $f$ that only depends on $\Delta$. 
    A related special case where $W\in \{0,1\}^{m\times n}$ has pairwise disjoint row supports was conjectured in \cite{AkramiLiuRajVegh2026}, with a discrepancy bound allowed to depend on $m$, and remains open. Such matrices have column sparsity $1$.
    \item Can we prove the Morell--Skutella conjecture with a discrepancy bound $2\cdot \max_{j\in [n]}d_j$?
\end{itemize}

\section*{Acknowledgement} 
The authors used GPT-5.6 Sol, GPT-5.6 Pro and Opus~5, including as part of a custom harness, during the development of this work to explore proof strategies and counterexamples. The proof strategies for the weighted equitability theorem and algorithm, as well as the counterexamples to the weighted carpooling conjecture and Morell--Skutella conjecture, were first discovered by AI and then simplified by the authors. Every AI-generated proof was verified and rewritten by the authors, who take full responsibility for the paper. 

This work grew out of a collaboration with Tam\'as Schwarcz and L\'aszl\'o V\'egh, and we are grateful for their valuable input. Krist\'of is grateful to Nicole Megow for initial discussions on the matroid-constrained load balancing problem. Part of the work was completed during an internship of the second author at Microsoft Research and visits to University of Bonn and Eötvös Lor\'and University. The research received further support from the Lend\"ulet Programme of the Hungarian Academy of Sciences (LP2021-1/2021), from the Ministry of Innovation and Technology of Hungary from the National Research, Development and Innovation Fund (ADVANCED 150556 and 153096), and from the Dynasnet European Research Council Synergy project (ERC-2018-SYG 810115).

\bibliographystyle{plain} 
\bibliography{reference} 

\appendix
\section{Proof of Theorem \ref{thm:weighted_degree_bounded_basis}}\label{sec:proof-degree-bases}
\begin{proof}[Proof of Theorem \ref{thm:weighted_degree_bounded_basis}]
    Let $\widebar{M}=(\widebar{E},\widebar{\mathcal{I}})\coloneqq M$, let $\widebar{W}\in \R_{\ge 0}^{\widebar{m}\times\widebar{E}}\coloneqq W$ and let $\widebar{b}\coloneqq Wx$. During the algorithm, by a slight abuse of notation, $M=(E,\mathcal{I})$ denotes the current matroid, $W\in \R_{\ge 0}^{m\times E}$ denotes the current constraint matrix, and $b$ denotes the current right-hand side. Initially, $b=\widebar{b}$. 
Then, $x$ is feasible for the following LP.
    \begin{equation}\label{eq:degree_matroid}
        \begin{aligned}
            \min&\ \sum_{e\in E} c_ey_e\\
            s.t.\ 
            &y(S)\leq r(S)\quad &\forall S\subseteq E\\
            &y(E)=r(E)\\
            &Wy=b\\
            &0\leq y_e\leq 1\quad &\forall e\in E.
        \end{aligned}
    \end{equation}
    Consider the following iterative rounding algorithm.
    \begin{algorithm}[htbp]
\caption{\textsc{Degree-bounded Bases}}
\label{alg:degree-bases}
\begin{algorithmic}[1]
\Require A matroid \(M=(E,\cI)\), weights $W\in \R_{\ge 0}^{m\times E}$ such that there are at most $\Delta$ nonzeros per column, costs $c\in \R^E$, and $x\in P(M)$.
\Ensure A basis $B$ such that $\inftynorm{W\ind_B-Wx}\leq 2\Delta\cdot \max_{i\in[m],e\in E} W_{i,e}$ and $c(B)\le c^\top x$.

\State Initialize with \(B\gets \emptyset\), and $b\gets Wx$.
\While{\(B\) is not a basis}
\State Let $y^*$ be an optimal extreme point of \eqref{eq:degree_matroid}. 
\State For every $e\in E$ such that $y^*_e=0$, delete the $e$-th column from $W$, and update $M\gets M\backslash e$. \label{step:a}
\State For every $e\in E$ such that $y^*_e=1$, delete the $e$-th column $w_e$ from $W$, let $b\gets b-w_e$, and update matroid $M\gets M/e$. Let $B\gets B\cup \{e\}$\label{step:b}
\State For every row $w_i^\top$ with at most $2\Delta$ nonzeros, delete $w_i^\top$ from $W$.\label{step:c}
\EndWhile

\State \Return \(B\)
\end{algorithmic}
\end{algorithm}

We first verify the discrepancy guarantee. Fix a row $\widebar{w}_i^\top$. Consider the extreme point solution $y^*$ right before the $i$-th constraint is dropped (if it is never dropped, we pick $y^*$ from the last iteration). Suppose $M=(E,\mathcal{I})$ is the current matroid, $w_i$ is the current row, and $b_i$ is the current right-hand side. Then,
\[
|\widebar{w}_i(B)-\widebar{b}_i|=\left|w_i(B\cap E)-w_i^\top y^*\right|=\left|\sum_{e\in E}w_i(e)(\ind_B(e)-y^*(e))\right|\le \sum_{e\in E}w_i(e)\le 2\Delta\cdot\max_{e\in E}w_i(e),
\]
where the first equality follows from steps \ref{step:a}, \ref{step:b} and the fact that $w_i^\top y^*=b_i$; the last inequality follows from step \ref{step:c}.
    
    Next, we prove the cost guarantee. After deleting an element with $y^*_e=0$ or contracting an element with $y^*_e=1$, the restriction of $y^*$ remains feasible for the next residual LP. Deleting a row from constraint matrix $W$ only enlarges the feasible region. Therefore, the cost of the elements already added to $B$ plus the optimum value of the residual LP never increases. Since $x$ is feasible for the initial LP, it follows that $c(B)\leq c^\top x$.

    We are left to prove that the algorithm terminates. Assume for the sake of contradiction that none of the conditions in steps \ref{step:a}, \ref{step:b}, and \ref{step:c} holds. Hence, $0<y^*_e<1$ for every $e\in E$, and the row $w_i^\top$ has at least $2\Delta+1$ nonzero entries for every $i$. Let $\mathcal{T}=\{S\subseteq E\colon y^*(S)=r(S)\}$ be the collection of all tight sets. By a standard uncrossing argument, there is a chain $\mathcal{L}=\{S_1,\dots,S_\ell\}$ where $S_1\subsetneq\cdots\subsetneq S_\ell$ such that $\spa(\{\chi_S\colon S\in\mathcal{L}\})=\spa(\{\chi_S\colon S\in\mathcal{T}\})$. Set $S_0\coloneqq\emptyset$. For every $h\in[\ell]$, one has $0<y^*(S_h\setminus S_{h-1})=y^*(S_h)-y^*(S_{h-1})=r(S_h)-r(S_{h-1})\in\Z$. Hence, $y^*(S_h\setminus S_{h-1})\geq 1$. Further, since $y^*_e<1$ for every $e\in S_h\setminus S_{h-1}$, we have $|S_h\setminus S_{h-1}|\geq 2$. Therefore, $|\mathcal{L}|\leq |E|/2$. 
    
    Let $m$ be the number of rows for the current matrix $W$. Since every column of $W$ has at most $\Delta$ nonzero entries and every row contains at least $2\Delta+1$ nonzero entries, 
    \[
        (2\Delta+1)m\leq |E|\Delta.
    \]
    Therefore, $m<|E|/2$, and hence $|\mathcal{L}|+m< |E|$. Since none of the bounds $0\leq y^*_e\leq 1$ is tight, the tight constraints of \eqref{eq:degree_matroid} span a space of dimension at most $|\mathcal{L}|+m<|E|$. This contradicts the fact that $y^*$ is an extreme point.
\end{proof}

\section{Proof of Lemma \ref{lem:prefix-cost}}\label{sec:prefix-cost-proof}
Let $Q\subseteq \R^n$ be a polyhedron, let $Z\subseteq Q$, and let $R\subseteq \R^n$ be a convex body containing the origin. A \emph{$(Q,Z)$-rounding algorithm} with error body $R$ takes $x\in Q$ as input and returns $y\in Z$ such that $y\in x+R$. The algorithm is \emph{face-preserving} if $y$ lies on the minimal face of $Q$ containing $x$. Swamy, Traub, Koch, and Zenklusen \cite{swamy2026unsplittable} prove the following:

\begin{theorem}[Swamy, Traub, Koch, and Zenklusen]\label{thm:cost-preserving}
    Assume there is a face-preserving $(Q,Z)$-rounding algorithm with error body $R$. Then, given costs $c\in \R_{\ge 0}^n$, there is an algorithm that receives an input $x\in Q$ and returns a solution $y\in Z$ satisfying $y\in x+(R-R)$ and $c^\top y\le c^\top x$.
\end{theorem}

We are ready to prove Lemma \ref{lem:prefix-cost}.
\begin{proof}[Proof of Lemma \ref{lem:prefix-cost}]
    Let $Q$ be the matroid base polytope $P(M)$, and let $Z\subseteq Q$ be the collection of indicator vectors of bases of $M$. Define the error body
    \[
    R\coloneqq \left\{z\in \R^n: \left|\sum_{i=1}^t w_iz_i\right|\le C\cdot \max_{e\in [n]} w_e\quad\forall t\in [n]\right\}.
    \]
    Assume Conjecture \ref{conj:matroid_chain} is true. Given $x\in Q$, let $F\subset Q$ be the minimal face containing $x$. Then, $F$ is the base polytope of another matroid $N$ on the same ground set $[n]$ (see e.g. \cite{fujishige1984characterization}). Applying Conjecture \ref{conj:matroid_chain} to matroid $N$ and weights $w$, we obtain a basis $B$ of $N$ satisfying \eqref{eq:prefix-bases}. In particular, $B$ lies on the face $F$. Let $y$ be the indicator vector of $B$. Notice that \eqref{eq:prefix-bases} is equivalent to $y-x\in R$. Therefore, we obtain a face-preserving $(Q,Z)$-rounding algorithm with error body $R$. It follows from Theorem \ref{thm:cost-preserving} that there is an algorithm that returns a solution $y\in Z$ satisfying $y\in x+(R-R)$ and $c^\top y\le c^\top x$. Since $R-R=2R$ by the central symmetry of $R$, we have $y-x\in 2R$, which is equivalent to
    \[\left|\sum_{i=1}^tw_i\mathbf{1}\{i\in B\}-\sum_{i=1}^tw_ix_i\right|\leq 2C\cdot \max_{e\in [n]} w_e\quad \forall t\in [n].\qedhere\] 
\end{proof}

\section{Proof of Proposition \ref{prop:binary-weights}}\label{sec:proof-binary}
\begin{proof}[Proof of Proposition \ref{prop:binary-weights}]
    Let $b_t\coloneqq \sum_{i=1}^t w_ix_i$. The LP relaxation of min-cost prefix-constrained bases is the following:
    \begin{equation}\label{eq:prefix_matroid}
        \begin{aligned}
            \min&\ \sum_{e\in [n]} c_ey_e\\
            s.t.\ 
            &y(S)\leq r(S)\quad &\forall S\subseteq [n]\\
            &y([n])=r([n])\\
            \rounddown{b_t}&\le \sum_{i=1}^tw_iy_i\le \roundup{b_t}\quad&\forall t\in [n]\\
            &0\leq y_e\leq 1\quad &\forall e\in [n].
        \end{aligned}
    \end{equation}
    Let $y^*$ be an extreme point optimal to \eqref{eq:prefix_matroid}. We assume $0<y^*_e<1$ for every $e\in [n]$ by fixing the variables $y_e^*\in\{0,1\}$ and considering a smaller instance. Let $\mathcal{T}=\{S\subseteq [n]\colon y^*(S)=r(S)\}$ be the collection of all tight sets from the matroid constraints. By a standard uncrossing argument, there is a chain $\mathcal{L}=\{S_1,\dots,S_\ell\}$ where $S_1\subsetneq\cdots\subsetneq S_\ell$ such that $\spa(\{\chi_S\colon S\in\mathcal{L}\})=\spa(\{\chi_S\colon S\in\mathcal{T}\})$. Let $R_t\coloneqq \{i\in [t]:w_i=1\}$. Then, $\mathcal{R}=\{R_1,\ldots, R_n\}$ forms another chain, where $R_1\subseteq \cdots\subseteq R_n$. The tight sets from the prefix constraints are defined by incidence vectors of sets in $\cal{R}$. Therefore, the constraint matrix defining $y^*$ forms two chains, which is totally unimodular (see e.g. \cite{frank2011connections} Section 4.2.2). Therefore, $y^*\in \{0,1\}^n$. Let $B\coloneqq \{e\in [n]:y^*_e=1\}$. Then, $B$ is a basis of $M$ satisfying 
    \[\left|\sum_{i=1}^tw_i\mathbf{1}\{i\in B\}-\sum_{i=1}^tw_ix_i\right|\leq 1\quad \forall t\in [n]. \qedhere\]
\end{proof}

\section{Proof of Theorem \ref{thm:uniform}}\label{sec:proof-uniform}

We prove the prefix-constrained basis conjecture for uniform matroids and arbitrary nonnegative weights.
\begin{proof}[Proof of Theorem \ref{thm:uniform}]
Let $r$ be the rank of $M$. We define a sequence of vectors $x = x^{(0)}, x^{(1)}, \dots, x^{(n)} \in [0,1]^n$ so that for each $t \in [n]$, the following five invariants are preserved:

\begin{itemize}
\item $x_i^{(t)} = x_i^{(0)}$ for all $i > t$;
\item $x^{(t)}$ has at most one fractional coordinate in $[t]$;
\item If $x^{(t)}_i \in \{0,1\}$ for some $i \in [t]$ then $x^{(t')}_i = x^{(t)}_i$ for every $t' > t$;
\item $\sum_{i=1}^t w_i x_i^{(t)} = \sum_{i=1}^t w_i x_i^{(0)}$;
\item $\sum_{i=1}^n x_i^{(t)} \in [r,r+1)$.
\end{itemize}

Given this sequence, we define $B \coloneqq  \{i : x^{(n)}_i = 1\}$. By the second invariant, $x^{(n)}$ has at most one fractional coordinate. By the fifth invariant, $\sum_{i=1}^n x_i^{(n)}\in[r,r+1)$. Therefore exactly $r$ coordinates of $x^{(n)}$ are equal to $1$, and hence $|B|=r$. Since $M$ is the rank-$r$ uniform matroid, $B$ is a basis. First we show that this satisfies the desired bound for every prefix $t \in [n]$. If $x^{(t)}$ has no fractional coordinate in $[t]$, $x^{(n)}_i = x^{(t)}_i$ for all $i \in [t]$ by the third invariant, so $$\left|\sum_{i=1}^tw_i\mathbf{1}\{i\in B\}-\sum_{i=1}^tw_ix_i\right| = \left|\sum_{i=1}^tw_i x^{(t)}_i -\sum_{i=1}^tw_i x^{(0)}_i\right| = 0.$$
Otherwise, let $p \in [t]$ be the unique fractional coordinate of $x^{(t)}$ in $[t]$. Then 
\[
\left|\sum_{i=1}^t w_i\mathbf{1}\{i\in B\}-\sum_{i=1}^t w_i x_i\right|
= w_p\left|\mathbf{1}\{p\in B\}-x_p^{(t)}\right|
\leq w_p
\leq \max_{e\in[n]}w_e.
\]

It remains to construct the sequence. Suppose that $x^{(t-1)}$ has been constructed and set $x^{(t)} = x^{(t-1)}$. If $x^{(t)}$ contains at most one fractional coordinate, no change is needed. Otherwise, there are exactly two: $t$ and some $p < t$. It suffices to replace $x^{(t)}_p, x^{(t)}_t$ by one of the two endpoints of the line segment $\{(u,v) \in [0,1]^2 : w_p u + w_t v = w_p x^{(t)}_p + w_t x^{(t)}_t\}$. Assume the two endpoints are $z_i=(u_i,v_i),\ i=1,2$. Each $z_i$ has at most one fractional coordinate, and thus the second invariant is preserved. Next, we prove that one of them satisfies the fifth invariant. We claim that the difference of their coordinate sum
$|(u_1+v_1)-(u_2+v_2)|$ is at most $1$.
Indeed, since $w \ge 0$, $u_1-u_2$ and $v_1-v_2$ has opposite signs. Therefore,
\[
|(u_1+v_1)-(u_2+v_2)|=\Big||u_1-u_2|-|v_1-v_2|\Big|\le \max\Big\{|u_1-u_2|,|v_1-v_2|\Big\}\le 1.
\]
Since $(x_p^{(t)},x_t^{(t)})$ is a convex combination of $z_1,z_2$ and originally $\sum_{i=1}^n x_i^{(t)} \in [r, r+1)$, one of the $z_i$ keeps the sum in this interval, as needed.  \end{proof}

\end{document}